\documentclass[12pt]{article}
\usepackage[right=1in,left=1in,top=1in,bottom=1in]{geometry}
\usepackage{hyperref}
\hypersetup{colorlinks, citecolor=blue, filecolor=blue, linkcolor=blue, urlcolor=blue}
\usepackage{graphicx}
\usepackage{url}
\usepackage[round]{natbib}
\usepackage{amsmath,amsthm}
\usepackage{engord}
\usepackage{float}
\usepackage{subfig}
\usepackage{pdflscape}
\usepackage{booktabs}
\usepackage{amssymb}
\usepackage{multirow}
\usepackage{xr}
\usepackage{xcolor}

\newtheorem{theorem}{Theorem}

\newtheorem{definition}{Definition}
\newtheorem{assumption}{Assumption}
\newtheorem{proposition}{Proposition}
\newtheorem{lemma}{Lemma}
\newtheorem{corollary}{Corollary}
\theoremstyle{definition}
\newtheorem{remark}{Remark}

\DeclareMathOperator{\Var}{Var}
\DeclareMathOperator{\Cov}{Cov}
\newcommand{\E}{\mathbb{E}}

\usepackage{setspace}
\renewcommand{\arraystretch}{1.5}

\usepackage{sectsty}
\sectionfont{\large}
\subsectionfont{\normalsize}
\subsubsectionfont{\normalsize}

\title{\vspace*{-2.5cm} \hspace*{-0.5cm}Identification and Estimation of Optimal Continuous Treatment Effects}

\author{Fangzhou Yu\thanks{School of Economics, University of Sydney. \href{mailto:fangzhou.yu@sydney.edu.au}{fangzhou.yu@sydney.edu.au}}}

\date{\vspace*{0.5cm}July, 2026}

\begin{document}

\bgroup
\let\footnoterule\relax

\begin{singlespace}
  \maketitle

  % ===== Long abstract (previous version) — commented out in favor of the short version below =====
  % \begin{abstract}
  %   Estimating causal effects of continuous treatments is challenging because traditional methods rely on unstable estimation of the score of the conditional treatment density. A recent paper by \citet{hines2023optimally} bypasses this instability by treating a bounded outcome weight, rather than a derivative weight, as the analytical primitive, thereby characterizing a class of weighted average derivative effects without density estimation. In this paper, we develop the identification and estimation theory for the optimally efficient estimands of this class under homoskedasticity and heteroskedasticity. On identification, we show that these estimands relax standard conditions, permitting selection on the basis of unobserved outcome variance, and remaining valid at sharp boundaries and at interior treatment deserts that the strict overlap and density-smoothness conditions of classical theory rule out. On estimation, we derive the efficient influence function, and develop Debiased Machine Learning estimators. Under homoskedasticity, the optimally efficient estimand coincides with the projection coefficient in a Partially Linear Regression; under heteroskedasticity, the optimal estimands center the treatment around a precision-weighted propensity score, generalizing the overlap-weighted effect of categorical treatments. We illustrate the proposed estimators through an application examining the effect of winning a lottery prize on labor supply.
  % \end{abstract}

  \begin{abstract}
    Estimating continuous treatment effects is hard because average-derivative estimators rely on an ill-posed conditional-density score. Recent work makes a bounded outcome weight the primitive, characterizing a class of weighted average derivative effects without density estimation. In this paper, we develop the identification and estimation theory for the optimally efficient estimands of this class under homoskedasticity and heteroskedasticity. On identification, we show that these estimands relax standard conditions, remaining valid at sharp boundaries and at interior treatment deserts that the strict overlap and density-smoothness conditions of classical theory rule out. On estimation, we derive the efficient influence function, and develop Debiased Machine Learning estimators.
  \end{abstract}



\end{singlespace}
\thispagestyle{empty}

\clearpage
\egroup
\setcounter{page}{1}

\section{Introduction}\label{sec:intro}

Estimating the causal effects of continuous treatments or exposures, such as drug dosage, the intensity of a policy intervention, or exposure to economic shocks, is a fundamental challenge in empirical economics. While the dose-response curve, \(d \mapsto \E[Y_i(d)]\), provides a complete characterization of the treatment effect of dose \(d\), nonparametrically estimating this entire function is a statistically demanding task typically characterized by slow, nonparametric convergence rates. In many empirical settings, researchers require a scalar summary of the marginal effect that can be estimated at the parametric root-\(n\) rate.

For continuous treatments, a natural scalar summary is a weighted Average Derivative Effect (ADE), taking the form \(\tau_w = \E[w(D_i, X_i)\mu'(D_i, X_i)]\), where \(D_i\) is the treatment and \(X_i\) is the set of covariates, \(\mu(d, x) = \E[Y_i \mid D_i=d, X_i=x]\) is the conditional expectation function, \(\mu'\) is its partial derivative with respect to \(d\), and \(w(d, x)\) is a user-specified weighting function defining the target population. However, the efficient estimation of the ADE also presents statistical challenges. From a functional analysis perspective, the target parameter \(\tau_w\) evaluates a differential operator applied to the unknown regression function \(\mu\). Consequently, the mapping \(\mu \mapsto \E[w \mu']\) is generally an unbounded linear functional, and thus, one cannot directly invoke the Riesz representation theorem in \(L_2(P_{D, X})\) to guarantee the existence of a stable, square-integrable representer.

Instead, traditional semiparametric estimation must construct a representation explicitly. This is typically achieved via integration by parts, which shifts the unbounded derivative operator from the unknown conditional mean \(\mu\) to the conditional density of the treatment \(f(d|x)\) \citep{powell1989semiparametric,newey1993efficiency}. This operation yields a weighting function on the observed outcome, \(\alpha_w(d, x) = -\partial_d w(d, x) - w(d, x)\partial_d \log f(d|x)\), such that \(\tau_w = \E[\alpha_w(D_i, X_i)Y_i]\). This method, while specifying the target population \(w\) and subsequently deriving its representer \(\alpha_w\), exposes a theoretical vulnerability. The explicitly constructed \(\alpha_w\) inherently relies on the score of the conditional density, \(\partial_d \log f(d|x)\). Nonparametrically estimating the derivative of a conditional density is an ill-posed inverse problem. It amplifies finite-sample noise, requires strong smoothness assumptions, and suffers from severe instability when the density is near zero \citep[e.g.,][]{newey1994large, cattaneo2010robust}.

This instability mirrors a well-known vulnerability in the categorical treatment literature, where standard Average Treatment Effect (ATE) estimators rely on inverse probability weights that blow up when propensity scores approach the boundaries. To resolve the limited overlap problem, researchers developed the ``balancing weights'' framework \citep{crump2006moving, li2018balancing,li2019propensity}. By shifting the estimand to prioritize subpopulations with substantial covariate overlap, this framework directly parameterizes weights that minimize the asymptotic variance of the estimator, thereby identifying inherently stable estimands and bypassing extreme propensity weights.

Recent work has extended covariate balancing principles to continuous treatments. For example, \citet{fong2018covariate} and \citet{ai2021unified} propose unified frameworks that directly estimate continuous stabilized weights, \(f(d)/f(d|x)\), by forcing covariates and treatments to be uncorrelated. While these methods bypass maximum likelihood density estimation, they are designed to recover causal dose-response levels \(\E[Y(d)]\). When researchers require a summary of the marginal effect, these level-based IPW frameworks generally restrict analysis to parameterized structural models (e.g., linear regressions) or discrete contrasts (e.g., \(\E[Y(d') - Y(d)]\)), instead of a fully nonparametric summary of the derivative effects.

This instability can be overcome by reversing the traditional analytical paradigm. Rather than specifying a derivative weight \(w\) and employing an unbounded differential operator to derive an unstable representer \(\alpha_w\), one defines the target estimand by directly specifying a bounded, square-integrable outcome weight \(\alpha \in L_2(P_{D,X})\). By restricting the analytical primitive to a bounded \(\alpha\) from the outset, the mapping \(\mu \mapsto \E[\alpha \mu]\) is a bounded, continuous linear functional by design, and the estimand is freed from any reliance on the conditional density score.

This outcome-weight primitive was recently formalized by \citet{hines2023optimally}. They develop a class \(\mathcal{R}\) of square-integrable functions satisfying \(\E[\alpha D]=1\) and \(\E[\alpha \mid X]=0\), and derive its optimally efficient members. Two features of their analysis motivate ours. First, on the analytical side, they treat the derivative mapping \(\mu \mapsto \E[w\mu']\) as a bounded linear functional, while we show that this differential operator is in general unbounded in the \(L_2(P_{D,X})\) topology, so that specifying \(\alpha\) as the primitive is not merely an algorithmic convenience but a necessity for a well-posed functional. Second, while they derive the general heteroskedastic optimal estimand, its Riesz representer depends on the conditional outcome variance. Judging this dependence to obscure the interpretation of the target population, they recommend two simpler least-squares estimands in the literature for applied use and do not develop an estimators for the optimal heteroskedastic estimands. Their identification, moreover, maintains full conditional independence between the potential outcomes and the treatment, together with density-smoothness conditions standard in the continuous treatment effect literature.

For the bounded functional \(\E[\alpha(D_i, X_i) Y_i]\) to identify a causally interpretable average derivative free from confounding bias, \(\alpha\) must be orthogonal to the pre-treatment covariates, i.e., \(\E[\alpha(D_i, X_i) \mid X_i] = 0\). This orthogonality condition serves as the continuous treatment generalization of the covariate balance required in the categorical treatment literature \citep{li2018balancing}. By treating \(\alpha\) as the primitive, the implied derivative weight \(w_\alpha(d, x)\) describes the aggregation scheme of the conditional causal derivatives, connecting directly to an active econometrics literature that evaluates the implicit treatment-effect weights of regression estimators \citep[e.g.,][]{blandhol2022tsls, goldsmith2022contamination,borusyak2024negative}.

We make two contributions that target the gaps left by this framework. Our first contribution is to develop the identification of the optimal estimands and to show that it systematically relaxes the standard conditions. Identification requires only conditional mean independence, which permits economic agents to sort into treatments on the basis of unobserved outcome variance (for example, risk-averse firms choosing lower investment because high investment carries volatile returns). This makes a general heteroskedastic estimator economically relevant. We further show that identification extends to settings with sharp economic thresholds (a zero lower bound on investment, bunching at a maximum tax bracket) and interior treatment deserts (dose regions with no observed treatments), both of which violate the density-positivity and density-smoothness conditions of classical average derivative theory. To make these relaxations precise, we express the estimand through a cumulative weight that never divides by the conditional density, and we read the optimal weights as an efficiency-driven reweighting toward the high signal-to-noise regions of the dose-response curve.

Our second contribution is to solve the general heteroskedastic estimation problem that \citet{hines2023optimally} derived but declined to estimate. We derive the Efficient Influence Function (EIF) of the optimal estimand, which features a debiased precision weight that tracks the squared outcome residuals so as to eliminate the first order bias of estimating the conditional variance. Because debiasing the precision weight injects higher-order terms into the Neyman-orthogonal remainder, the standard Cauchy--Schwarz bilinear argument no longer delivers \(\sqrt{n}\)-consistency. We provide two sets of conditions with the uniform bounds and envelope conditions required to control these remainders and restore \(\sqrt{n}\) inference. As a by-product, the analysis shows that under homoskedasticity, the optimal estimand coincides with the population projection coefficient of a Partially Linear Regression (PLR) \citep{robinson1988root}, supplying an efficiency-based causal interpretation of that widely used coefficient even under heterogeneous effects.

The remainder of the paper is organized as follows. Section~\ref{sec:wade} formulates balancing weights for continuous treatments and establishes the non-negativity guarantee. Section~\ref{sec:optim} derives the optimal balancing weights and develops the relaxed causal identification of the corresponding estimands. Section~\ref{sec:est} derives the efficient influence functions and establishes the higher-order DML theory for estimation. Section~\ref{sec:simulation} reports a Monte Carlo study, and Section~\ref{sec:app} illustrates the estimators through an application to the effect of lottery winnings on labor supply. Section~\ref{sec:conclusion} concludes. All proofs are collected in the Appendix.

\section{Balancing Weights for Continuous Treatments}\label{sec:wade}

Suppose we observe an independent and identically distributed sample of \(n\) observations \(O_i = (Y_i, D_i, X_i)\) drawn from an unknown joint distribution \(P\). Here, \(Y_i \in \mathbb{R}\) is the observed outcome, \(D_i \in \mathcal{D} \subseteq \mathbb{R}\) is a continuous treatment or exposure, and \(X_i \in \mathbb{R}^p\) is a vector of pre-treatment covariates. We adopt the potential outcomes framework, letting \(Y_i(d)\) denote the potential outcome for unit \(i\) under treatment level \(d \in \mathcal{D}\).

While the dose-response curve \(d \mapsto \E[Y_i(d)]\) provides a full characterization of the treatment effect, estimating this entire function nonparametrically is a statistically demanding task that suffers from the curse of dimensionality. In many empirical settings, researchers require a single scalar summary of the marginal effect estimable at the parametric \(\sqrt{n}\) rate.

To formalize a causal target parameter, we define the weighted average causal derivative as the limit of an expected incremental policy shift
\[
  \dot{\tau}_w = \lim_{\nu \to 0} \E\left[ w(D_i, X_i) \frac{Y_i(D_i + \nu) - Y_i(D_i)}{\nu} \right],
\]
where \(w(d,x)\) is a weighting function. The causal parameter \(\dot{\tau}_w\) characterizes the effect of increasing the treatment level by \(\nu\) on the expected potential outcome \citep[e.g.,][]{kreif2015evaluation}.

\subsection{Identification and the Statistical Target}\label{sec:iden}

To connect the causal parameter \(\dot{\tau}_w\) to the observable data distribution \(P\), let \(\mu(d,x) = \E[Y_i \mid D_i=d, X_i=x]\) denote the observed conditional expectation function, with its partial derivative with respect to \(d\) denoted as \(\mu'(d, x) = \partial_d \mu(d, x)\). Let \(f(d|x)\) denote the conditional density of \(D_i\) given \(X_i = x\), and let \(m(d, x) = \E[Y_i(d) \mid X_i = x]\) denote the conditional expected potential outcome.

Under the standard causal identification conditions in the continuous setting, collected as Assumption \ref{ass:iden} in Appendix \ref{sec:appenda}\footnote{Note that our identification Assumption \ref{ass:iden} is weaker than those provided in \citet{hines2023optimally}. See discussions in Appendix \ref{sec:appenda}}, the weighted average causal derivative is identified by the average derivative of the observed conditional mean,
\[
  \dot{\tau}_w = \E[w(D_i, X_i)\mu'(D_i, X_i)],
\]
as recorded in Proposition \ref{prp:iden} of Appendix \ref{sec:appenda}. The corresponding result for a binary treatment \(D_i \in \{0, 1\}\) is \(\E[w(X_i)(Y_i(1) - Y_i(0))] = \E[w(X_i)(\mu(1, X_i) - \mu(0, X_i))]\), as developed in the balancing weights literature \citep{crump2006moving, li2018balancing}. We adopt these classical conditions only to fix the statistical target. Once attention is restricted to the optimal estimands, Section \ref{sec:id} shows that they can be substantially relaxed, in particular dispensing with the local overlap condition on \(f(d|x)\).

In light of this identification, we drop the dot notation for the remainder of the paper and focus on the statistical estimand \(\tau_w \equiv \E[w(D_i, X_i)\mu'(D_i, X_i)]\), keeping in mind its structural justification as the continuous limit of a causal shift intervention.

\subsection{The Instability of Average Derivatives}\label{sec:instability}

Let \(\mathcal{H} = L_2(P_{D,X})\) denote the Hilbert space of all measurable functions \(g: \mathcal{D} \times \mathbb{R}^p \to \mathbb{R}\) that are square-integrable with respect to the joint distribution of \((D_i, X_i)\). The parameter \(\tau_w\) evaluates a linear functional of the conditional response surface
\[
  T_w(\mu) = \E[w(D_i, X_i)\mu'(D_i, X_i)].
\]
In modern semiparametric theory, achieving \(\sqrt{n}\)-consistent estimation of a linear functional fundamentally relies on finding its Riesz representer, a function \(\alpha_w \in \mathcal{H}\) such that \(T_w(\mu) = \langle \alpha_w, \mu \rangle = \E[\alpha_w(D_i, X_i)\mu(D_i, X_i)]\). By the law of iterated expectations, the existence of such an \(\alpha_w\) allows the estimand to be expressed directly as a weighted average of the observed outcome \(\tau_w = \E[\alpha_w(D_i, X_i)Y_i]\).

However, a mathematical obstacle arises. The functional \(T_w\) applies a differential operator, \(\mu \mapsto \mu'\). In the standard \(L_2(P_{D,X})\) topology, differential operators are strictly unbounded linear functionals. \citet{hines2023optimally} establish boundedness of this mapping under explicit regularity conditions on \(w\) and the data-generating process (their Conditions 1--3), and thereby invoke the Riesz Representation Theorem; absent such conditions, however, the differential operator is in general unbounded, and a unique, square-integrable representer is not guaranteed. Specifying \(\alpha\) as the primitive secures boundedness without imposing those conditions on the unknown data-generating process.

Because the Riesz Representation Theorem does not apply, traditional econometric approaches construct a representation explicitly. Under specific boundary and density smoothness assumptions (Assumption \ref{ass:powell} in Appendix \ref{sec:appenda}), one can apply integration by parts to shift the derivative operator from the unknown outcome function \(\mu\) to the conditional treatment density \(f(d|x)\) \citep{powell1989semiparametric},
\begin{equation}
  \label{eq:rr}
  \alpha_w(d,x) = -\frac{\partial w(d,x)}{\partial d} - w(d,x)\frac{\partial \log f(d|x)}{\partial d}.
\end{equation}
If the weighting scheme \(w\) and the underlying data-generating process are sufficiently regular to ensure this explicitly constructed \(\alpha_w\) has finite variance (\(\alpha_w \in \mathcal{H}\)), the functional is bounded ex post, and \(\alpha_w\) serves algebraically as its representer.

Equation \eqref{eq:rr} illustrates the theoretical vulnerability of estimating continuous treatment effects. Because the original functional is unbounded, deriving its representer forces a structural reliance on the score of the conditional density, \(\partial_d \log f(d|x)\). Consequently, the traditional analytical method, which fixes the derivative weight \(w\) as the primitive and subsequently maps \(w \mapsto \alpha_w\), demands restrictive smoothness conditions on the unobservable data-generating process. Even when these conditions hold, nonparametrically estimating the derivative of a conditional density is an ill-posed inverse problem. It amplifies finite-sample noise and causes estimators to become unstable when \(f(d|x)\) is close to zero, acting as a continuous analogue to the limited overlap problem in binary treatments \citep[e.g.,][]{cattaneo2010robust}.

\subsection{Generalized Balancing Weights}\label{sec:bw}

The instability inherent in Equation \eqref{eq:rr} arises because the standard approach specifies the target population via the derivative weight \(w\), and then employs an unbounded differential operator to derive the representer \(\alpha_w\). We propose to invert this analytical sequence. Rather than specifying \(w\), we define the statistical estimand by directly parameterizing an outcome weight \(\alpha \in \mathcal{H}\) as our analytical primitive.

By restricting our primitive \(\alpha\) to bounded elements in \(\mathcal{H}\) from the beginning, the resulting linear functional \(T_\alpha(\mu) = \E[\alpha(D_i, X_i)\mu(D_i, X_i)]\) is guaranteed to be bounded and continuous by the Cauchy-Schwarz inequality \(|T_\alpha(\mu)| \le \|\alpha\|_{\mathcal{H}}\|\mu\|_{\mathcal{H}}\). Under this inverted formulation, \(\alpha\) serves as the Riesz representer by design.

However, for this bounded functional \(\E[\alpha(D_i, X_i)Y_i]\) to identify a causally meaningful average derivative effect, \(\alpha\) cannot be an arbitrary function in \(L_2(P_{D,X})\). It must satisfy econometric constraints to purge confounding bias and correctly scale the treatment variation. We formalize these requirements in the following definition.
\begin{definition}[Continuous Balancing Weights]
  \label{def:class}
  A function \(\alpha \in L_2(P_{D,X})\) is defined as a continuous balancing weight if it satisfies two structural moment conditions: (i) Covariate Balance: \(\E[\alpha(D_i, X_i) \mid X_i] = 0\) almost surely; (ii) Normalization: \(\E[\alpha(D_i, X_i) D_i] = 1\).
\end{definition}

The two structural conditions in Definition \ref{def:class} coincide with the class \(\mathcal{R}\) of square-integrable weights studied by \citet{hines2023optimally}; our development emphasizes their econometric content and the causal interpretation of the resulting estimands. Condition (i) serves a dual purpose. Statistically, it ensures that the boundary terms vanish during integration by parts, isolating the pure derivative effect. Causally, it generalizes the concept of covariate balance established in the categorical treatment literature. For categorical treatments, balancing weights are defined by their ability to equalize the weighted covariate distributions across treatment regimes \citep{li2018balancing}. Condition (i) generalizes this distributional covariate balance to a continuum of treatment levels. Condition (ii) ensures the estimand correctly captures the magnitude of the marginal effect, which implies \(\E[w(D_i, X_i)] = 1\). We formally show these arguments in the proof of the following proposition.

Let \(F(d|x)\) denote the conditional cumulative distribution function (CDF) of \(D_i\) given \(X_i=x\).
\begin{proposition}
  \label{prp:class}
  For any continuous balancing weight in Definition \ref{def:class}, the statistical estimand \(\tau_\alpha = \E[\alpha(D_i, X_i)Y_i]\) identically recovers a normalized weighted average causal derivative, i.e., \(\E[w_\alpha(D_i, X_i)] = 1\), with the implied derivative weight uniquely defined as
  \begin{equation}
    \label{eq:weight}
    w_\alpha(d,x) = -\frac{F(d|x)}{f(d|x)} \cdot \E\big[\alpha(D_i, X_i) \mid D_i \le d, X_i = x\big].
  \end{equation}
\end{proposition}

The mapping in Proposition \ref{prp:class} coincides with Theorem 1 of \citet{hines2023optimally}. Their implied derivative weight reduces to the closed form \eqref{eq:weight} once the covariate balance constraint \(\E[\alpha \mid X_i]=0\) is imposed. We derive this form directly by integration by parts, without solving the differential equation for \(w\) on which their construction relies, and we retain it because it is the basis of the non-negativity characterization developed next.

Finally, a standard requirement for causal interpretation is that the implied derivative weights are non-negative. Non-negative weighting prevents the artificial ``sign reversal'' of local treatment effects, related to recent econometric discussion on design-based regressions \citep{blandhol2022tsls, goldsmith2022contamination, borusyak2024negative}. Proposition \ref{prp:class} allows us to characterize the necessary and sufficient condition on the balancing weight \(\alpha\) to guarantee non-negative \(w_\alpha(d,x)\).
\begin{lemma}
  \label{lem:nonneg}
  For a continuous balancing weight \(\alpha\) and its implied weight \(w_\alpha\) in Equation \eqref{eq:weight}, \(w_\alpha(d, x) \ge 0\) if and only if \(\E[\alpha(D_i, X_i) \mid D_i \le d, X_i = x] \le 0\) for all \(d\) and \(x\).
\end{lemma}

Intuitively, for the truncated expectation to be non-positive at the lower boundary of the treatment support, \(\alpha(d, x)\) must generally take negative values for small \(d\). Conversely, for the overall expectation \(\E[\alpha(D_i, X_i) \mid X_i]\) to return to exactly zero at the upper boundary, \(\alpha(d,x)\) must transition to positive values for large \(d\). Thus, Lemma \ref{lem:nonneg} enforces a structural requirement that a valid causal balancing weight \(\alpha(d,x)\) must transition from negative to positive as treatment intensity increases. While this condition provides an explicit diagnostic check for weight non-negativity, we treat the conditions in Definition \ref{def:class} as the sole restriction in Section \ref{sec:optim}, since the variance-minimizing solutions we develop inherently satisfy Lemma \ref{lem:nonneg}. As shown following Theorem \ref{thm:optimal_weights}, this holds for all of our optimal estimands, including the general heteroskedastic one. The guarantee is stronger than the sufficient condition of \citet{hines2023optimally}, who certify non-negative aggregation only when \(\alpha\) is monotonically increasing in \(d\). Because their optimal heteroskedastic weight, which is the same object as ours, need not be monotone, their condition fails to certify its non-negativity.

\section{Optimal Estimands and Relaxed Causal Identification}\label{sec:optim}

\subsection{Optimization Criterion and Regularity Conditions}\label{sec:crit}

The class of balancing weights is large, and Proposition \ref{prp:class} confirms that any continuous balancing weight identifies a normalized effect \(\tau_\alpha\). This raises the natural question: which balancing weight should be selected? Following the established approach in the balancing weights literature \citep{crump2006moving, li2019propensity}, we seek the balancing weight that is optimal in the sense of minimizing the nonparametric efficiency bound.

The efficiency bound characterizes the lowest possible asymptotic variance achievable by any regular asymptotically linear (RAL) estimator. When the derivative weight \(w(d,x)\) is known, the influence function of \(T_w\) has been derived by \citet{newey1993efficiency}:
\begin{equation}
  \label{eq:inf_func}
  \psi(O_i) = \underbrace{\alpha_w(D_i, X_i)\big(Y_i - \mu(D_i, X_i)\big)}_{\equiv \psi_A(O_i)} + \underbrace{w(D_i, X_i)\mu'(D_i, X_i) - \tau_w}_{\equiv \psi_B(O_i)}.
\end{equation}
Given the influence function and an efficient estimator \(\hat{\tau}_w\), its asymptotic variance is given by
\begin{equation}
  \label{eq:var_decomp}
  V_\psi = \E[\psi^2(O_i)] = \E[\psi_A^2(O_i)] + \E[\psi_B^2(O_i)],
\end{equation}
as \(\E[\psi_A(O_i)\psi_B(O_i)] = 0\), which is a consequence of the orthogonality between the residual term \(Y_i - \mu(D_i, X_i)\) and the conditional derivative effects. Here, \(\E[\psi_A^2(O_i)]\) represents the variance contribution from the noise in the outcome, amplified by the magnitude of the balancing weights \(\alpha_w\). \(\E[\psi_B^2(O_i)]\) is a measure of effect heterogeneity, representing the variance of the weighted individual derivative effects around the population average \(\tau_w\).

While minimizing the total variance \(V_\psi\) might seem intuitive, selecting \(w\) by minimizing the heterogeneity component \(\E[\psi_B^2(O_i)]\) is conceptually and practically problematic. Because \(\E[\psi_B^2(O_i)]\) is a function of the unknown population parameter \(\tau_w\), optimizing it leads to a circular definition. Furthermore, minimizing effect heterogeneity changes the objective from seeking statistical precision to contorting the targeted causal parameter based on the variation in treatment effects.

To optimize solely for statistical efficiency, we follow the strategy in \citet{crump2006moving} and redefine the objective as minimizing the asymptotic variance of the estimator relative to the sample analogue of \(\tau_w\), denoted \(\tau_{w,S}\). The asymptotic variance of \(\sqrt{n}(\hat{\tau}_w - \tau_{w,S})\) is the expectation of the square of the leading term \(\psi_A(O_i)\), which gives the conditional efficiency bound:
\begin{equation}
  \label{eq:vs}
  V_S = \E[\psi_A^2(O_i)] = \E[\alpha_w^2(D_i, X_i)\big(Y_i - \mu(D_i, X_i)\big)^2].
\end{equation}
Let \(\sigma^2(D_i, X_i) = \Var(Y_i \mid D_i, X_i)\) denote the conditional variance of the outcome. By the law of iterated expectations, the objective function in Equation \eqref{eq:vs} simplifies to \(V_S = \E[\alpha_w^2(D_i, X_i) \sigma^2(D_i, X_i)]\).

The solution to this functional optimization problem depends on the form of heteroskedasticity. To formalize a unified framework across all heteroskedastic architectures, we define the generalized precision weight as
\begin{equation}
  \label{eq:precision_weight}
  \omega(D_i, X_i) \equiv \sigma^{-2}(D_i, X_i).
\end{equation}
To ensure the optimization is well-defined and that the resulting optimal estimands possess square-integrable influence functions, we assume the following regularity conditions.

\begin{assumption}
  \label{ass:regularity_eif}
  \begin{itemize}
    \item[(a)] There exist strictly positive constants \(c_1, c_2\) such that \(c_1 \le \Var(Y_i \mid D_i, X_i) \le c_2\) almost surely. Consequently, the precision weight is bounded by \(1/c_2 \le \omega(D_i, X_i) \le 1/c_1\).
    \item[(b)] There exists \(c_3 > 0\) such that \(\E\big[\Var(D_i \mid X_i)\big] \ge c_3\).
    \item[(c)] The treatment and outcome have finite second moments, \(\E[D_i^2] < \infty\) and \(\E[Y_i^2] < \infty\).
  \end{itemize}
\end{assumption}

\begin{remark}
  Assumption \ref{ass:regularity_eif} plays a role distinct from the causal identification conditions developed in Section \ref{sec:id}. The identification conditions deliver the causal interpretation of the target, whereas Assumption \ref{ass:regularity_eif} is a separate regularity requirement that guarantees the optimization is well-posed and the optimal weights have finite variance. Assumption \ref{ass:regularity_eif}(a) bounds the precision weight away from zero and infinity, and Assumption \ref{ass:regularity_eif}(b) requires the precision-weighted conditional variance of the treatment to be bounded away from zero, so that the denominator of the optimal estimand does not degenerate. Assumption \ref{ass:regularity_eif}(c) supplies the finite second moments of \(D_i\) and \(Y_i\) under which the normalization \(\E[\alpha D]\) and the estimand \(\tau^*_\omega\) are well-defined. Together they ensure that the optimal balancing weight \(\alpha^*_\omega\) lies in \(L_2(P_{D,X})\). Additional moment conditions required for \(\sqrt{n}\)-consistent estimation are introduced separately in Theorem \ref{thm:dml_asymptotics}.
\end{remark}

\subsection{Optimal Continuous Balancing Weights}\label{sec:optbw}

We select the optimal balancing weight \(\alpha^*\) as the unique function that minimizes the conditional variance bound \(V_S\) subject to the balancing constraints in Definition \ref{def:class}.

By solving this constrained functional optimization problem in \(L_2(P_{D,X})\), we obtain a unified theorem that characterizes the optimal estimands across different heteroskedasticity architectures. To express the general solution, we define the precision-weighted propensity as
\begin{equation}
  \label{eq:precision_weighted_mean}
  e_\omega(x) = \frac{\E[\omega(D_i, X_i) D_i \mid X_i=x]}{\E[\omega(D_i, X_i) \mid X_i=x]}.
\end{equation}
Note that when the outcome variance is conditionally independent of the treatment, \(e_\omega(x)\) collapses to the standard propensity \(e(x) = \E[D_i \mid X_i=x]\).

\begin{theorem}[Optimal Continuous Balancing Weights]
  \label{thm:optimal_weights}
  Suppose Assumption \ref{ass:regularity_eif} holds. The optimal balancing weight \(\alpha^*\) that minimizes the efficiency bound \(V_S\), and its corresponding optimal estimand \(\tau^*\), are uniquely determined by the structure of the precision weight \(\omega\)
  \begin{align}
    \alpha^*_\omega(d,x) & = \frac{\omega(d,x)\big(d - e_\omega(x)\big)}{\E\big[\omega(D_i, X_i) \big(D_i - e_\omega(X_i)\big)^2\big]}, \label{eq:alpha_unified}                         \\
    \tau^*_\omega        & = \frac{\E\big[\omega(D_i, X_i) \big(D_i - e_\omega(X_i)\big) Y_i\big]}{\E\big[\omega(D_i, X_i) \big(D_i - e_\omega(X_i)\big)^2\big]}. \label{eq:tau_unified}
  \end{align}
  \begin{itemize}
    \item[(a)] General Heteroskedasticity: If \(\sigma^2(D_i, X_i)\) varies with both \(D_i\) and \(X_i\), the precision weight is \(\omega(d,x) = \sigma^{-2}(d,x)\).
    \item[(b)] Covariate-Dependent Heteroskedasticity: If the conditional variance depends only on pre-treatment covariates, \(\omega(d,x) = \sigma^{-2}(x)\). The precision-weighted propensity simplifies to \(e(x)\).
    \item[(c)] Homoskedasticity: If the outcome is homoskedastic, \(\omega(d,x) = 1\). The optimal balancing weight reduces to \(\alpha^*(d,x) = (d - e(x)) / \E[\Var(D_i \mid X_i)]\), yielding the optimal estimand \(\tau^* = \E[\Cov(D_i, Y_i \mid X_i)] / \E[\Var(D_i \mid X_i)]\).
  \end{itemize}
\end{theorem}

Theorem \ref{thm:optimal_weights} recovers, in unified precision-weight notation, the variance-minimizing representers characterized by \citet{hines2023optimally}. Case (a) is their Theorem 2, and case (c) coincides with their Corollary 1, the partially linear estimand \(\Psi\). Case (b) specializes Theorem 2 to covariate-dependent heteroskedasticity within the globally normalized class of Definition \ref{def:class}, and is distinct from their within-stratum-optimal estimand \(\psi\), as detailed in Section \ref{sec:challenge}.

Theorem \ref{thm:optimal_weights} provides a characterization of efficiency in continuous treatment settings, yielding three key insights. First, under general heteroskedasticity, the structural noise level fluctuates along the continuous treatment path itself. In this scenario, orthogonalizing against the standard propensity \(e(X_i)\) is no longer theoretically optimal. Instead, efficiency dictates that the treatment distribution must be centered around \(e_\omega(X_i)\), which serves as the optimal predictor of \(D_i\) under a precision-weighted inner product. This localized centering guarantees that the estimand heavily weights the specific segments of the dose-response curve that exhibit the highest signal-to-noise ratio.

This continuous optimization framework unifies established optimal weights from the discrete treatment literature. Consider the binary treatment setting where \(D_i \in \{0, 1\}\). Let \(e(X_i) = \mathbb{P}(D_i=1 \mid X_i)\), \(\sigma^2_1(X_i) = \Var(Y_i \mid D_i=1, X_i)\), and \(\sigma^2_0(X_i) = \Var(Y_i \mid D_i=0, X_i)\). The precision-weighted propensity becomes
\[
  e_\omega(X_i) = \frac{e(X_i)/\sigma^2_1(X_i)}{e(X_i)/\sigma^2_1(X_i) + (1-e(X_i))/\sigma^2_0(X_i)}.
\]
The target population weight implied by the optimal continuous balancing weight for treated units is \(w^*(X_i) = \alpha^*_\omega(1, X_i)e(X_i)\). Simplification of Theorem \ref{thm:optimal_weights}(a) under this discrete support reveals
\[
  w^*(X_i) \propto \left( \frac{\sigma^2_1(X_i)}{e(X_i)} + \frac{\sigma^2_0(X_i)}{1-e(X_i)} \right)^{-1}.
\]
This equivalence recovers the optimal heteroskedastic overlap weights established in \citet{crump2006moving}.

Second, when the noise varies across covariates but not across doses, the optimal estimand is analogous to Feasible Generalized Least Squares (FGLS). It inversely weights subpopulations by their outcome variance \(\sigma^2(X_i)\), giving more weight to covariate strata with more precisely measured outcomes.

Third, under homoskedasticity, the optimal continuous estimand \(\tau^*\) is algebraically equivalent to the population projection parameter from a Partially Linear Regression (PLR) \citep{robinson1988root}. In empirical practice, the PLR coefficient is often justified as a causal parameter only under the restrictive assumption of constant treatment effects. This equivalence establishes an efficiency-based causal interpretation for the PLR parameter even when treatment effects are heterogeneous: it identifies the weighted average derivative effect that minimizes the nonparametric efficiency bound. Furthermore, adapting the optimization criterion under homoskedasticity to minimize the sum of the asymptotic variances for all pairwise comparisons recovers the Generalized Overlap Weights \citep{li2019propensity} (See Appendix \ref{sec:appendb}).

\begin{remark}
  All three of our optimal estimands satisfy the non-negativity property by Lemma \ref{lem:nonneg}. From Theorem \ref{thm:optimal_weights}, \(\alpha_\omega^*(d,x)\) is proportional to \(\omega(d,x)(d - e_\omega(x))\) scaled by a positive denominator. We evaluate the sign of the unscaled truncated expectation \(I(d,x) = \E[\omega(D_i,X_i)(D_i - e_\omega(x))\mathbb{I}(D_i \le d) \mid X_i=x]\) in two regions. For \(d \le e_\omega(x)\), the term \((D_i - e_\omega(x))\) is non-positive almost surely on the event \(\{D_i \le d\}\). Since the precision weight is strictly positive \(\omega(D_i,X_i) > 0\), it follows that \(I(d,x) \le 0\). For \(d > e_\omega(x)\), we utilize the covariate balance constraint. By the definition of \(e_\omega(x)\), the full conditional expectation is zero, meaning \(\E[\omega(D_i,X_i)(D_i - e_\omega(x)) \mid X_i=x] = 0\). We can therefore rewrite the truncated expectation as the negation of the upper tail \(I(d,x) = -\E[\omega(D_i,X_i)(D_i - e_\omega(x))\mathbb{I}(D_i > d) \mid X_i=x]\). On the event \(\{D_i > d\}\), we have \(D_i > d > e_\omega(x)\), making the term inside the expectation positive. Its negation is therefore negative, yielding \(I(d,x) < 0\). Therefore, \(\E[\alpha_\omega^*(D_i,X_i) \mid D_i \le d, X_i=x] \le 0\) for all \(d\) and \(x\).
\end{remark}

\subsection{Relaxed Causal Identification}\label{sec:id}

Theorem \ref{thm:optimal_weights} characterizes the optimal estimands as functionals of the observed distribution. We now ask under what conditions \(\tau^*_\omega = \E[\alpha^*_\omega(D_i, X_i) Y_i]\) identifies a weighted average causal derivative. The classical conditions of Section \ref{sec:iden} suffice, but they are stronger than necessary; for the optimal estimands, identification holds under a markedly weaker set of conditions, and we isolate exactly which classical requirements can be dropped.

The key device is a cumulative reformulation of the implied weight that never divides by the conditional density. For a balancing weight \(\alpha\), define the cumulative weight
\begin{equation}
  \label{eq:cumulative}
  W_\alpha(d, x) = -\E\big[\alpha(D_i, X_i)\,\mathbb{I}(D_i \le d) \mid X_i = x\big].
\end{equation}
This is the same object as \(H(d,x) = w_\alpha(d,x) f(d|x)\) that arises in the proof of Proposition \ref{prp:class}. Writing it directly as a conditional expectation, rather than as the density-weighted derivative weight \(w_\alpha = W_\alpha / f\), makes two properties transparent. First, \(W_\alpha(\cdot, x)\) is well-defined and of bounded variation for almost every \(x\) whenever \(\alpha(\cdot, x) \in L_1(P_{D|X})\), which \(\alpha \in L_2(P_{D,X})\) supplies; it never divides by \(f(d|x)\) and remains finite where the conditional density vanishes or has atoms. Second, by the covariate balance constraint \(\E[\alpha \mid X]=0\), it vanishes at both ends of the support: \(W_\alpha \to 0\) at the lower limit trivially, and at the upper limit because \(W_\alpha(\bar{d}, x) = -\E[\alpha \mid X=x] = 0\).

We collect the minimal conditions under which the optimal estimands retain a causal interpretation.
\begin{assumption}[Causal Identification of the Optimal Estimands]
  \label{ass:id_min}
  For a continuous balancing weight \(\alpha\) in Definition \ref{def:class}:
  \begin{enumerate}
    \item[(a)] Consistency: \(D_i = d\) implies \(Y_i = Y_i(d)\).
    \item[(b)] Mean independence: \(\E[Y_i(d) \mid D_i, X_i] = \E[Y_i(d) \mid X_i] \equiv m(d, X_i)\) for all \(d \in \mathcal{D}\).
    \item[(c)] Absolute continuity: for almost every \(x\), \(m(\cdot, x)\) is absolutely continuous on the convex hull \([\underline{d}(x), \overline{d}(x)]\) of the support of \(D_i \mid X_i = x\); in particular it is continuous with finite limits at the boundary.
    \item[(d)] Dominated derivative: there is a measurable envelope \(L(d,x)\) with \(|\partial_d m(d,x)| \le L(d,x)\) and \(\E\big[\int |W_\alpha(d, X_i)|\, L(d, X_i)\, \mathrm{d}d\big] < \infty\), where \(\partial_d m\) is the almost-everywhere derivative supplied by (c).
  \end{enumerate}
\end{assumption}

Relative to the classical Assumption \ref{ass:iden}, three requirements are weakened or removed. Mean independence in (b) is retained but, as discussed below, is weaker than the full conditional independence assumed in the related literature. The smoothness requirement is relaxed from continuous differentiability to absolute continuity in (c); continuity alone secures the integration by parts, while absolute continuity is what licenses writing the resulting Stieltjes integral against the almost-everywhere derivative. Most importantly, the local overlap condition that \(f(d|x)\) be continuous and strictly positive in a neighborhood of every observed point is dropped entirely, because \(W_\alpha\) is defined through the conditional law and never divides by \(f\). Condition (d) is the integrability requirement that the classical Assumption \ref{ass:iden}(e) supplied in a different form. It guarantees the derivative representation below is finite and enables the interchange of integration. We retain the regularity Assumption \ref{ass:regularity_eif}, as it guarantees \(\alpha^*_\omega \in L_2(P_{D,X})\), so that the optimal weight from Theorem \ref{thm:optimal_weights} exists, and it implies the non-degeneracy \(\E[\omega(D_i - e_\omega(X_i))^2] > 0\) of the estimand's denominator.
\begin{proposition}[Self-Regularizing Identification]
  \label{prp:id_min}
  Suppose Assumptions \ref{ass:regularity_eif} and \ref{ass:id_min} hold for a continuous balancing weight \(\alpha\). Then \(\alpha\), and in particular each optimal weight \(\alpha^*_\omega\) of Theorem \ref{thm:optimal_weights}, identifies a weighted average causal derivative,
  \begin{equation}
    \label{eq:id_repr}
    \tau_\alpha = \E[\alpha(D_i, X_i) Y_i] = \E\left[ \int_{\underline{d}(X_i)}^{\overline{d}(X_i)} \frac{\partial m(d, X_i)}{\partial d}\, W_\alpha(d, X_i)\, \mathrm{d}d \right],
  \end{equation}
  where \([\underline{d}(x), \overline{d}(x)]\) is the convex hull of the support of \(D_i \mid X_i = x\). For the optimal weights, \(W_{\alpha^*_\omega} \ge 0\) by Lemma \ref{lem:nonneg}.
\end{proposition}

The proof, in Appendix \ref{sec:appenda}, applies integration by parts to \(\E[\alpha Y] = \E[\alpha m]\). The boundary term \([-m\,W_\alpha]\) vanishes because \(W_\alpha \to 0\) at both endpoints by covariate balance and \(m\) is bounded there by Assumption \ref{ass:id_min}(c), and the form accommodates atoms in the treatment distribution without a density. On the observed support, \(\partial_d m\) aggregates local causal derivatives. Across an interior gap in the support, \(W_\alpha\) is constant, so the gap contributes \(m(\overline{a}, x) - m(\underline{a}, x)\) between its endpoints. The pointwise derivative enters \(\tau_\alpha\) only through this identified endpoint contrast, under the maintained continuity of \(m\). The framework thus does not eliminate extrapolation so much as relocate it from the conditional density, which classical methods must smooth and invert, to the outcome regression \(m\). Integrating the cumulative weight recovers the normalization, \(\int_{\underline{d}(x)}^{\overline{d}(x)} W_\alpha(d, x)\, \mathrm{d}d = \E[\alpha(D_i, X_i) D_i \mid X_i = x]\), so \(\E\big[\int W_\alpha\, \mathrm{d}d\big] = \E[\alpha D] = 1\); the representation \eqref{eq:id_repr} thus expresses \(\tau_\alpha\) as a normalized average of the local derivatives \(\partial_d m\), weighted non-negatively for each optimal \(\alpha^*_\omega\) by Lemma \ref{lem:nonneg}.

We note that there are two contrasts with \citet{hines2023optimally}. The first concerns selection on variance. \citet{hines2023optimally} and the classical average-derivative literature assume full conditional independence \(Y_i(d) \perp\!\!\!\perp D_i \mid X_i\); Assumption \ref{ass:id_min}(b) requires only conditional mean independence, which permits agents to sort into dosage on the basis of unobserved outcome variance, while leaving the conditional mean unconfounded. Admitting such selection is what makes a precision-weighted estimand economically meaningful rather than a purely statistical refinement.

The second concerns boundary robustness. The vulnerability of the classical approach is that it makes the boundary term vanish through the conditional density itself: the integration by parts underlying the unweighted average derivative carries the boundary term \([\mu(d,x)\,f(d|x)]\), discarded only if \(f(d|x)\) decays at the edges of the support \citep{powell1989semiparametric}. Openness of the support does not secure this. A treatment that is uniform, \(D_i \mid X_i \sim \mathcal{U}(0, c(X_i))\), has open support \((0, c(x))\) yet a density that limits to \(1/c(x) \neq 0\) at both ends, so the boundary term does not vanish and the average derivative is not identified without further restrictions on the weight. Covariate balance removes this dependence on the density's tail behavior. The boundary term is \([-m(d,x)\,W_\alpha(d,x)]\), and \(W_\alpha \to 0\) at both endpoints (trivially at the lower limit, and at the upper limit because \(W_\alpha(\overline{d}, x) = -\E[\alpha(D_i, X_i) \mid X_i = x] = 0\) by Definition \ref{def:class}), so it vanishes regardless of the value of \(f\). To make the boundary term vanish, \citet{hines2023optimally} instead impose two regularity conditions on the treatment density, that its support is an open interval, and that the density-weighted weight \(w(d,x)\,f(d|x)\) is differentiable on it. The cumulative weight \(W_\alpha\) requires neither. Defined directly from the conditional distribution of \(D\), it is automatically of bounded variation, which is all the integration by parts needs. It admits closed boundaries and boundary mass points and shifts the smoothness requirement from the treatment density onto the outcome regression \(m\). Specifying \(\alpha\) as the primitive is in this sense well matched to economic applications with sharp boundaries and mass points in the treatment.

\subsection{Efficiency-Driven Regularization}\label{sec:reg}

The optimal weight \(\alpha^*_\omega(d,x) \propto \omega(d,x)(d - e_\omega(x))\) reweights the dose-response curve toward regions with a high signal-to-noise ratio. This reweighting is driven by the efficiency criterion of Theorem \ref{thm:optimal_weights}, and here, we describe how it specializes to the three estimands.

\emph{Homoskedastic estimand}. Under \(\omega = 1\), the implied weight is proportional to \(\Var(D_i \mid X_i)\), covariate strata in which the treatment is more variable receive more weight. This is the continuous analogue of the variance weighting that partial regression places on covariate cells, and at the binary boundary, where \(\Var(D_i \mid X_i) = e(X_i)(1 - e(X_i))\), it reduces exactly to the categorical overlap weight \citep{crump2006moving, li2018balancing}. A covariate stratum with a degenerate treatment, \(\Var(D_i \mid X_i) = 0\), receives exactly zero weight, so the estimand resolves limited-overlap strata without ad hoc trimming, while strata with small but positive treatment variance are down-weighted.

\emph{Covariate-dependent estimand}. Under \(\omega = \sigma^{-2}(x)\), the estimand additionally reweights covariate strata by the inverse of their outcome variance, exactly as Feasible Generalized Least Squares does, giving more weight to strata with more precisely measured outcomes. Under Assumption \ref{ass:regularity_eif}(a) the precision weight is bounded, so this reweighting is bounded and shifts weight toward higher-precision subpopulations.

\emph{General heteroskedastic estimand}. Under \(\omega = \sigma^{-2}(d,x)\), the inverse-variance weighting operates within the dose-response curve, shifting weight toward dose regions where the outcome is measured with less noise. Under Assumption \ref{ass:regularity_eif}(a) this too is a bounded reweighting. It is tempting to read this as protection against confounding in the tails, where extreme doses are often driven by idiosyncratic shocks. Suppose that mean independence is violated, and let \(b(d,x) = \mu(d,x) - m(d,x)\) denote the resulting confounding gap. The estimand then decomposes as \(\tau^*_\omega = \E[\alpha^*_\omega m] + \E[\alpha^*_\omega b]\), the causal target plus a bias \(\E[\alpha^*_\omega(D_i,X_i) b(D_i,X_i)]\). Because each region enters this bias with weight \(\alpha^*_\omega \propto \omega(d,x)(d - e_\omega(x)) = \sigma^{-2}(d,x)(d - e_\omega(x))\), a confounding gap is down-weighted in proportion to the local outcome variance, so regions of high outcome noise, where extreme-dose confounding is most likely to reside, enter the bias with diminished weight. If one is willing to posit that confounding is collocated with noise, \(|b(d,x)| \le c\,\sigma^2(d,x)\), the bias is bounded by \(c\,\E[|D_i - e_\omega(X_i)|] / \E[\omega(D_i - e_\omega(X_i))^2]\).

\section{Semiparametric Estimation via DML}\label{sec:est}

\subsection{The Challenge of Heteroskedasticity}\label{sec:challenge}

The optimal estimands of Theorem \ref{thm:optimal_weights} minimize the efficiency bound, yet only the homoskedastic estimand is straightforward to estimate. Under \(\omega = 1\) the precision weight is constant and \(\mu\) cancels from the estimating equation, which collapses to the partially linear regression score; standard Debiased Machine Learning then delivers \(\sqrt{n}\)-consistency under the usual \(o_P(n^{-1/4})\) nuisance rates \citep{robinson1988root, chernozhukov2018double}. This is the regime to which the prior literature confined itself, recommending the partially linear and conditionally homoskedastic estimands for applied use.\footnote{The two estimands recommended by \citet{hines2023optimally} are the partially linear coefficient \(\Psi = \E[\Cov(D_i,Y_i \mid X_i)]/\E[\Var(D_i \mid X_i)]\), which coincides with our \(\tau^*_{\mathrm{homo}}\), and \(\psi = \E[\Cov(D_i,Y_i \mid X_i)/\Var(D_i \mid X_i)]\), which is optimal under conditional homoskedasticity \(\sigma^2(d,x)=\sigma^2(x)\) within the subclass of weights normalized inside each covariate stratum, \(\E[\alpha(D_i,X_i) D_i \mid X_i]=1\). The estimand \(\psi\) is distinct from our covariate-dependent estimand \(\tau^*_{\mathrm{cov}}\) of Theorem \ref{thm:optimal_weights}(b): although both pertain to the same variance structure \(\sigma^2(d,x)=\sigma^2(x)\), the representer of \(\psi\) is \((D_i-e(X_i))/\Var(D_i \mid X_i)\), which depends only on the treatment model \(P_{D \mid X}\) and requires no outcome-variance estimation, whereas \(\tau^*_{\mathrm{cov}}\) is optimal within the globally normalized class of Definition \ref{def:class} (\(\E[\alpha(D_i,X_i) D_i]=1\)), reweights covariate strata by \(\sigma^{-2}(X_i)\Var(D_i \mid X_i)\), and does require estimating \(\sigma^2(x)\). We work throughout in the globally normalized class and do not study the within-stratum-normalized subclass.} The general heteroskedastic estimand is the difficult case, and it is the one we solve. Its weight \(\omega(d,x) = \sigma^{-2}(d,x)\) depends on the conditional outcome variance, an additional nonparametric nuisance function; a naive plug-in that substitutes a machine-learning estimate \(\hat{\sigma}^2\) into the weights inherits its first-order regularization bias, which destroys parametric convergence. Neutralizing this bias requires an influence function orthogonal to the estimation of \(\sigma^2\), and the resulting score is no longer bilinear in the nuisance errors. The next two subsections derive that influence function and the higher-order machine-learning theory it demands.

\subsection{The Optimal Efficient Influence Function}\label{sec:eif}

Theorem \ref{thm:optimal_weights} characterizes the optimal continuous estimands assuming the structural components, such as \(\omega\), \(e_\omega\), and \(\mu\), are known. In practice, these nonparametric nuisance parameters must be estimated. To estimate the optimal estimands at the parametric \(\sqrt{n}\)-rate using flexible machine learning methods, modern semiparametric theory requires us to find the Efficient Influence Function (EIF) of the target parameter. The EIF automatically provides the Neyman orthogonal score necessary to immunize the estimator against the slow convergence rates of the nuisance parameters \citep{chernozhukov2018double}.

To analyze the EIF, we first construct a projection representation shared by all three optimal estimands. We define the precision-weighted conditional outcome mean \(\rho_\omega(x)\) as
\begin{equation}
  \label{eq:rho_omega}
  \rho_\omega(x) = \frac{\E[\omega(D_i, X_i) \mu(D_i, X_i) \mid X_i = x]}{\E[\omega(D_i, X_i) \mid X_i = x]}.
\end{equation}
Applying the law of iterated expectations, all three optimal estimands \(\tau^*_\omega\) defined in Theorem \ref{thm:optimal_weights} can be written as a projection
\begin{equation}
  \label{eq:unified_projection}
  \tau^*_\omega = \frac{\E\big[\omega(D_i, X_i)\big(D_i - e_\omega(X_i)\big)\big(\mu(D_i, X_i) - \rho_\omega(X_i)\big)\big]}{\E\big[\omega(D_i, X_i)\big(D_i - e_\omega(X_i)\big)^2\big]}.
\end{equation}

We now derive the EIF for \(\tau^*_\omega\), taking \(\omega = \sigma^{-2}(D_i, X_i)\) to be the \((D_i, X_i)\)-conditional precision that defines the general heteroskedastic estimand of Theorem \ref{thm:optimal_weights}(a); the covariate-dependent and homoskedastic estimands follow as specializations under their maintained variance architectures, as Remark \ref{rmk:eif_scope} details. Let \(\eta = (\omega, \mu, e_\omega, \rho_\omega)\) denote the vector of nuisance parameters. Because the precision weight \(\omega = \sigma^{-2}\) depends on the conditional variance of the outcome, estimating it introduces a specific adjustment into the influence function. To capture this analytically, we define the debiased precision weight as
\begin{equation}
  \label{eq:omega_db}
  \omega_{\mathrm{db}}(O_i; \eta) = 2\omega(D_i, X_i) - \omega^2(D_i, X_i)\big(Y_i - \mu(D_i, X_i)\big)^2.
\end{equation}
Notice that when evaluated at the true parameters \(\eta_0\), the conditional expectation satisfies \(\E[\omega_{\mathrm{db}}(O_i; \eta_0) \mid D_i, X_i] = 2\omega_0 - \omega_0^2 \sigma_0^2 = \omega_0\), where \(\sigma_0^2(D_i, X_i) = \E[(Y_i - \mu_0(D_i, X_i))^2 \mid D_i, X_i]\) is the \((D_i, X_i)\)-conditional variance and \(\omega_0 = \sigma_0^{-2}\). This calibration holds precisely because \(\omega_0\) is the \((D_i, X_i)\)-conditional precision: the conditional expectation of the debiased weight exactly recovers the true precision weight, while its tracking of the squared outcome residuals ensures Neyman orthogonality.

\begin{theorem}[Efficient Influence Function of Optimal Estimands]
  \label{thm:eif_optimal}
  Suppose Assumption \ref{ass:regularity_eif} holds, and let \(\eta_0\) denote the true nuisance parameters. The Efficient Influence Function (EIF) of \(\tau^*_\omega\) is given by \(\phi^*(O_i) = K_{\omega_0}^{-1}\psi_{\mathrm{aug}}(O_i; \tau^*_\omega, \eta_0)\), where \(K_{\omega_0} = \E[\omega_0(D_i, X_i)(D_i - e_{\omega,0}(X_i))^2]\), and the unscaled augmented score is \(\psi_{\mathrm{aug}}(O_i; \tau, \eta) = A_{\mathrm{aug}}(O_i; \eta) - \tau B_{\mathrm{aug}}(O_i; \eta)\), with components
  \begin{align}
    A_{\mathrm{aug}}(O_i; \eta) & = \omega_{\mathrm{db}}(O_i; \eta)\big(D_i - e_\omega(X_i)\big)\big(\mu(D_i, X_i) - \rho_\omega(X_i)\big) \nonumber \\
                                & \qquad \qquad + \omega(D_i, X_i)\big(D_i - e_\omega(X_i)\big)\big(Y_i - \mu(D_i, X_i)\big), \label{eq:score_A}     \\
    B_{\mathrm{aug}}(O_i; \eta) & = \omega_{\mathrm{db}}(O_i; \eta)\big(D_i - e_\omega(X_i)\big)^2. \label{eq:score_B}
  \end{align}
\end{theorem}

Theorem \ref{thm:eif_optimal} characterizes the statistical limit for estimating continuous treatment effects and highlights three theoretical properties. First, the EIF natively bypasses the ill-posed density derivative estimation (\(\partial_d \log f(d|x)\)). Second, the score is insensitive to the estimation of nuisance parameters \(e_\omega\) and \(\rho_\omega\) because the estimand is defined via an orthogonal projection. Finally, the appearance of the debiased weight \(\omega_{\mathrm{db}}\) corrects for the first-order bias induced by estimating the structural precision \(\omega(d,x)\).

\begin{remark}[Scope of the debiasing and variance-architecture robustness]
  \label{rmk:eif_scope}
  The debiased weight \(\omega_{\mathrm{db}}\) is constructed for the general heteroskedastic estimand, where \(\omega_0 = \sigma_0^{-2}(D_i, X_i)\) is the \((D_i, X_i)\)-conditional precision; the calibration \(\E[\omega_{\mathrm{db}} \mid D_i, X_i] = \omega_0\) holds because \(\omega_0\) is that conditional object. The other two estimands are specializations. Under homoskedasticity \(\omega \equiv 1\) is known and carries no first-order estimation error, so the debiasing is unnecessary and the score reduces to the partially linear regression score \eqref{eq:plr_score}. Under covariate-dependent heteroskedasticity the same calibration holds provided the maintained restriction \(\sigma_0^2(d,x) = \sigma_0^2(x)\) is correct. This yields a robustness distinction worth making explicit. Because the general implementation regresses the squared residuals on \((D_i, X_i)\), it recovers the correct conditional precision whatever the true variance architecture, and so consistently targets the optimal estimand even when the noise is in fact covariate-only or constant. A restricted implementation that models the variance as an \(X_i\)-measurable object \(\sigma_b^2(x)\), as \(\hat{\tau}^*_{\mathrm{cov}}\) does, is not architecture-agnostic: if the truth varies with the dose, its debiased weight calibrates not to \(\omega_b(x) = \sigma_b^{-2}(x)\) but to the pseudo-weight \(2\omega_b(x) - \omega_b^2(x)\sigma_0^2(d,x)\), which is not sign-definite and turns negative wherever \(\sigma_0^2(d,x) > 2\sigma_b^2(x)\). The estimator then targets a pseudo-parameter and the non-negativity of Lemma \ref{lem:nonneg} may fail. This reflects a misspecified variance architecture rather than a defect of the estimand, and it accounts for the finite-sample fragility of \(\hat{\tau}^*_{\mathrm{cov}}\) under the general-heteroskedastic design in Section \ref{sec:simulation}.
\end{remark}

\subsection{Debiased Machine Learning}\label{sec:dml}

The optimal continuous balancing weights derived in Theorem \ref{thm:optimal_weights} identify a class of target estimands \(\tau^*_\omega\) that minimize the nonparametric efficiency bound under various heteroskedastic scenarios. In practice, evaluating these estimands requires estimating the nonparametric nuisance functions \(\eta = (\omega, \mu, e_\omega, \rho_\omega)\). Replacing these unknown functions with flexible machine learning estimates in standard plug-in estimators generally fails to achieve parametric \(\sqrt{n}\)-consistency due to regularization bias and slow convergence rates. To overcome this, we exploit the EIF derived in Theorem \ref{thm:eif_optimal} and develop \(\sqrt{n}\)-consistent estimators based on the DML framework.

The estimation procedure is described as follows:
\begin{enumerate}
  \item Randomly partition the sample indices \(\{1, \dots, n\}\) into \(K\) folds, \(\mathcal{I}_1, \dots, \mathcal{I}_K\). Let \(\mathcal{I}_k^c\) denote the out-of-fold complement for fold \(k\).

  \item For each fold \(k \in \{1, \dots, K\}\), estimate the nuisance parameters.
        \begin{itemize}
          \item[2a.] Train machine learning models to estimate the conditional mean \(\hat{\mu}_k(d,x)\) by regressing \(Y_i\) on \((D_i, X_i)\). Then, estimate the conditional variance \(\hat{\sigma}^2_k(d,x)\) by regressing the squared residuals \((Y_i - \hat{\mu}_k(D_i, X_i))^2\) on \((D_i, X_i)\). Construct the precision weight as \(\hat{\omega}_k(d,x) = 1/\max\{\hat{\sigma}^2_k(d,x), c\}\) for some small trimming constant \(c>0\). Because the squared residuals are heavy-tailed, an over-flexible variance regression can drive \(\hat{\sigma}^2_k\) spuriously small in sparse regions and inflate the precision weight; this is most pronounced for the covariate-dependent estimand, whose variance regression conditions on \(X_i\) alone. We therefore select the leaf size of the random-forest variance learner by out-of-bag error, which adapts its smoothness to the data so the estimated weights do not develop a spurious heavy tail.
          \item[2b.] Estimate \(\hat{e}_{\omega,k}(x)\) and \(\hat{\rho}_{\omega,k}(x)\) by passing \(\hat{\omega}_k(D_i, X_i)\) to the ML models trained on \(\mathcal{I}_k^c\):
                \begin{align*}
                  \hat{e}_{\omega,k}    & = \arg\min_{f \in \mathcal{F}} \sum_{j \in \mathcal{I}_k^c} \hat{\omega}_k(D_j, X_j) \big(D_j - f(X_j)\big)^2, \\
                  \hat{\rho}_{\omega,k} & = \arg\min_{g \in \mathcal{G}} \sum_{j \in \mathcal{I}_k^c} \hat{\omega}_k(D_j, X_j) \big(Y_j - g(X_j)\big)^2.
                \end{align*}
        \end{itemize}

  \item Construct the final cross-fitted estimator \(\hat{\tau}^*_\omega\) by solving the empirical moment condition \(\frac{1}{n}\sum_{k=1}^K \sum_{i \in \mathcal{I}_k} \psi_{\mathrm{aug}}(O_i; \hat{\tau}^*_\omega, \hat{\eta}_k) = 0\), which yields:
\end{enumerate}
\begin{equation}
  \label{eq:dml_estimator}
  \hat{\tau}^*_\omega = \frac{\sum_{k=1}^K \sum_{i \in \mathcal{I}_k} A_{\mathrm{aug}}(O_i; \hat{\eta}_k)}{\sum_{k=1}^K \sum_{i \in \mathcal{I}_k} B_{\mathrm{aug}}(O_i; \hat{\eta}_k)}.
\end{equation}

\begin{theorem}[Asymptotic Normality]
  \label{thm:dml_asymptotics}
  Suppose the conditions in Assumption \ref{ass:regularity_eif} hold, and let \(\eta_0 = (\omega_0, \mu_0, e_{\omega,0}, \rho_{\omega,0})\) denote the true nuisance parameters, and let \(\hat{\eta}_k = (\hat{\omega}_k, \hat{\mu}_k, \hat{e}_{\omega,k}, \hat{\rho}_{\omega,k})\) denote their cross-fitted machine learning estimators. Assume that at least one of the following two conditions holds:

  \begin{itemize}
    \item Condition A:
          \begin{enumerate}
            \item[(i)] There exists \(C>0\) such that \(|Y_i| \le C\) and \(|D_i| \le C\) almost surely.
            \item[(ii)] \(\|\hat{\mu}_k\|_\infty\), \(\|\hat{e}_{\omega,k}\|_\infty\), \(\|\hat{\rho}_{\omega,k}\|_\infty \le C\).
            \item[(iii)] \(\|\hat{\mu}_k - \mu_0\|_{P,2}\), \(\|\hat{\omega}_k - \omega_0\|_{P,2}\), \(\|\hat{e}_{\omega,k} - e_{\omega,0}\|_{P,2}\), and \(\|\hat{\rho}_{\omega,k} - \rho_{\omega,0}\|_{P,2}\) are all \(o_P(n^{-1/4})\).
          \end{enumerate}
    \item Condition B:
          \begin{enumerate}
            \item[(i)] There exists \(q \ge 8\) such that \(\E[|Y_i|^q] < \infty\) and \(\E[|D_i|^q] < \infty\).
            \item[(ii)] \(\|\hat{\mu}_k - \mu_0\|_\infty\), \(\|\hat{\omega}_k - \omega_0\|_\infty\), \(\|\hat{e}_{\omega,k} - e_{\omega,0}\|_\infty\), and \(\|\hat{\rho}_{\omega,k} - \rho_{\omega,0}\|_\infty\) are all \(o_P(1)\).
            \item[(iii)] Let \(v(O_i) = 1 + |D_i - e_{\omega,0}(X_i)| + |\mu_0(D_i,X_i) - \rho_{\omega,0}(X_i)|\). \(\|(\hat{\mu}_k - \mu_0)v(O)\|_{P,2}\), \(\|(\hat{\omega}_k - \omega_0)v(O)\|_{P,2}\), \(\|\hat{e}_{\omega,k} - e_{\omega,0}\|_{P,2}\), and \(\|\hat{\rho}_{\omega,k} - \rho_{\omega,0}\|_{P,2}\) are all \(o_P(n^{-1/4})\).
          \end{enumerate}
  \end{itemize}

  Then, the estimator in Equation \eqref{eq:dml_estimator} is a Regular Asymptotically Linear (RAL) estimator of \(\tau^*_\omega\), i.e.,
  \[
    \sqrt{n}(\hat{\tau}^*_\omega - \tau^*_\omega) \xrightarrow{d} \mathcal{N}(0, V^*_\omega),
  \]
  where
  \[
    V^*_\omega = \frac{\E[\psi_{\mathrm{aug}}^2(O_i; \tau^*_\omega, \eta_0)]}{(\E[\omega_0(D_i, X_i)(D_i - e_{\omega,0}(X_i))^2])^2},
  \]
  which can be consistently estimated by plugging in the sample analogues.
\end{theorem}

\begin{remark}
  Standard DML, and the homoskedastic estimand of the prior literature, rest on a bilinear remainder---a product of two nuisance errors, such as \(\Delta\mu \times \Delta e\)---which the Cauchy--Schwarz inequality bounds at \(o_P(n^{-1/2})\) under \(o_P(n^{-1/4})\) nuisance rates. The convergence requirements in Theorem \ref{thm:dml_asymptotics} are stronger because estimating the heteroskedastic optimal estimand requires the precision weight \(\omega(d,x)\), and debiasing the precision weight must track squared outcome residuals. This introduces higher-order polynomial terms into the Neyman orthogonal remainder (e.g., \((\hat{\mu}_k-\mu_0)^2 (\hat{e}_k-e_{\omega,0})^2\)), for which the standard Cauchy--Schwarz bilinear bounds do not apply. Theorem \ref{thm:dml_asymptotics} provides two sets of conditions to handle the higher-order terms. For variables with bounded support (Condition A), the complex remainders are uniformly bounded, relaxing the requirement back to the standard \(L_2(P)\) rate of \(o_P(n^{-1/4})\). Condition A(ii) can be satisfied by trimming the estimators or by utilizing natively bounded ML estimators such as tree-based ensembles. For unbounded variables (Condition B), we control the remainders by requiring uniform consistency (\(\|\hat{\eta}_k - \eta_0\|_\infty = o_P(1)\)) alongside \(L_2(P)\) rate of \(o_P(n^{-1/4})\) for nuisance estimators weighted by an "envelope function" \(v(O_i)\). Condition B(i) guarantees that the envelope \(v(O_i)\) possesses finite 4th moments, so its norm \(\|v(O)\|_{P,4}\) is a bounded constant. Consequently, this requirement reduces to the underlying nuisance estimators achieving the \(o_P(n^{-1/4})\) convergence rate.
\end{remark}

\begin{remark}
  We note that the conditions in Theorem \ref{thm:dml_asymptotics} are required for the estimation and inference under general heteroskedasticity. However, the conditions required under homoskedasticity are remarkably relaxed. When \(\omega = 1\), \(\mu(d,x)\) cancels from the augmented score, reducing to the classic Partially Linear Regression (PLR) score
  \begin{equation}
    \label{eq:plr_score}
    \psi_{\mathrm{homo}}(O_i; \tau, \eta) = \big(D_i - e(X_i)\big)\big(Y_i - \rho(X_i)\big) - \tau \big(D_i - e(X_i)\big)^2.
  \end{equation}
  The estimation of PLR has been intensively studied in the literature \citep[e.g.,][]{robinson1988root,chernozhukov2018double}. Because its remainder is bilinear, \(\sqrt{n}\)-consistency requires only standard \(o_P(n^{-1/4})\) rate for nuisance estimators, recovering Theorem 3.1 of \citet{chernozhukov2018double}. Our continuous balancing framework formalizes a causal justification for deploying the PLR estimator even when the true treatment effect is heterogeneous. Theorem \ref{thm:optimal_weights}(c) and Proposition \ref{prp:class} prove that the PLR coefficient identifies a well-defined, variance-weighted average effect that prioritizes subpopulations with high \(\Var(D_i \mid X_i)\). This estimand also guarantees non-negative weights via Lemma \ref{lem:nonneg}. Thus, researchers can safely trade the asymptotic efficiency of the general estimand for relaxed statistical requirements and finite-sample robustness by using PLR.
\end{remark}

\section{Simulation Study}\label{sec:simulation}

In this section, we conduct a Monte Carlo simulation study to evaluate the finite-sample performance of the proposed continuous balancing estimators. We draw a four-dimensional pre-treatment covariate vector, \(X_i = (X_{i1}, X_{i2}, X_{i3}, X_{i4})'\), from a standard uniform distribution, \(\mathcal{U}[0, 1]^4\). Note that \(X_{i4}\) is a noise variable that does not enter the data generating process.

The data generating process is
\begin{align*}
  D_i \mid X_i  & \sim \mathcal{B}\text{eta}\big(2 + X_{i1}, 2 + X_{i2}\big),                                             \\
  \mu(D_i, X_i) & = D_i - 0.5 D_i^2 + X_{i1}D_i + X_{i2} + X_{i3}^2,                                                      \\
  Y_i           & = \mu(D_i, X_i) + \sigma(D_i, X_i)\varepsilon_i, \quad \varepsilon_i \sim \mathcal{TN}_{[-3, 3]}(0, 1),
\end{align*}
where the idiosyncratic error \(\varepsilon_i\) is drawn from a truncated normal distribution on the interval \([-3, 3]\). The true marginal derivative effect is heterogeneous across both dosage and covariates, given by \(\mu'(D_i, X_i) = 1 - D_i + X_{i1}\). To evaluate the efficiency bounds derived in Theorem \ref{thm:optimal_weights}, we create three heteroskedastic scenarios:
\begin{align*}
  \text{Homoskedastic:} \quad              & \sigma(D_i, X_i) = 0.5                                \\
  \text{Covariate-Dependent:} \quad        & \sigma(D_i, X_i) = 0.5\sqrt{1 + 10X_{i1}^2}           \\
  \text{General Heteroskedasticity:} \quad & \sigma(D_i, X_i) = 0.5\sqrt{1 + 10X_{i1}^2 + 20D_i^2}
\end{align*}

Because the marginal treatment effect is heterogeneous, the theoretically optimal estimands identify distinct population parameters, so the bias and coverage of each estimator are calculated relative to its own corresponding estimand. We compute each target by numerical integration against the closed-form conditional law of \(D_i \mid X_i\) rather than by simulation: \(\tau^*_{\mathrm{homo}} = 0.986\) in all three designs, since the outcome noise enters its plim only through a mean-zero term, while the precision-weighted estimands coincide with it under homoskedasticity, where the precision weight is constant, equal \(\tau^*_{\mathrm{cov}} = \tau^*_{\mathrm{gen}} = 0.804\) under covariate-dependent heteroskedasticity, and equal \(\tau^*_{\mathrm{gen}} = 0.948\) under general heteroskedasticity, where the misspecified \(X\)-measurable weight of \(\hat{\tau}^*_{\mathrm{cov}}\) targets a pseudo-estimand of \(1.001\) instead (Remark \ref{rmk:eif_scope}). The within-stratum benchmark \(\hat{\psi}\) introduced below, and the average derivative targeted by the baseline estimator, each equal unity in this design and are likewise assessed against their own estimands.

We compare five estimators. The first three are our proposed optimal balancing estimands: the homoskedastic PLR estimand \(\hat{\tau}_{\mathrm{homo}}^*\), the covariate-dependent optimal estimand \(\hat{\tau}_{\mathrm{cov}}^*\), and the general optimal estimand \(\hat{\tau}_{\mathrm{gen}}^*\). They are implemented via the estimation procedure in Section \ref{sec:est} with 5-fold cross-fitting and trimming constant \(c = 0.01\), employing random forests (ranger, 500 trees) for all nuisance functions; following Section \ref{sec:dml}, the leaf size of the covariate-dependent estimand's variance regression is selected by out-of-bag error. The fourth estimator is the within-stratum least-squares estimand \(\hat{\psi}\) of \citet{hines2023optimally}, a benchmark that summarizes the within-stratum slopes without estimating the outcome variance, computed by a cross-fitted one-step correction. The central comparison is between the precision-weighted \(\hat{\tau}_{\mathrm{cov}}^*\) and \(\hat{\tau}_{\mathrm{gen}}^*\) and the homoskedastic \(\hat{\tau}_{\mathrm{homo}}^*\), the partially linear estimand recommended for applied use by the prior literature; the heteroskedastic designs are chosen to exercise the estimation of the variance-dependent weights that distinguish them.

The fifth estimator is the baseline Average Derivative Estimator (ADE) implemented via the Generalized Propensity Score (GPS) method by \citet{hiranok2004propensity}. We estimate the GPS by using a series estimator to obtain the conditional density \(f(d|x)\), and then estimate the outcome's conditional expectation via a quadratic polynomial regression. We take the analytical total derivative of the fitted dose-response function with respect to the dose, evaluate it at each observation's own \((D_i, X_i)\), and average over the joint empirical distribution to construct \(\hat{\tau}_{\mathrm{ADE}} = \frac{1}{n}\sum_i \hat{\mu}'(D_i, X_i)\). The standard error is calculated using the bootstrap method with 1,000 bootstrap samples.

We conduct 1,000 Monte Carlo simulations for sample sizes \(n = 500\), \(n = 2000\), and \(n = 5000\). Table \ref{tbl:sim_results} reports the mean estimate, standard error, bias, and coverage probability of the 95\% confidence intervals for each estimator; the reported standard error is the Monte Carlo standard deviation of the point estimates across replications. In the design for which it is optimal, each of the three balancing estimands is approximately unbiased for its target, covers near the nominal rate, and attains a markedly smaller standard error than the baseline GPS estimator; the within-stratum benchmark \(\hat{\psi}\) is likewise approximately unbiased throughout. The GPS estimator, by contrast, carries a large finite-sample bias of about \(-0.10\) at every sample size, and its coverage deteriorates sharply as the sample grows and this bias comes to dominate its shrinking standard error. Among the balancing estimands the ranking tracks the efficiency theory of Theorem \ref{thm:optimal_weights}: in each design, and at every sample size, the estimand that is optimal for that variance structure attains the smallest Monte Carlo standard error. Under homoskedasticity there is no heteroskedasticity for precision-weighting to exploit, so the homoskedastic PLR estimand \(\hat{\tau}_{\mathrm{homo}}^*\), which alone avoids estimating a variance function, is the most precise (\(0.038\) at \(n = 5000\)). Under covariate-dependent heteroskedasticity the covariate-dependent estimand \(\hat{\tau}_{\mathrm{cov}}^*\) is smallest, with a standard error roughly twelve to sixteen percent below that of \(\hat{\tau}_{\mathrm{homo}}^*\) at the two larger sample sizes (\(0.065\) versus \(0.074\) at \(n = 5000\)). Under general heteroskedasticity only the fully optimal \(\hat{\tau}_{\mathrm{gen}}^*\), whose weight depends on the dose, captures the dose-dependent component of the noise, and it is the most precise at every sample size (\(0.114\) versus \(0.118\) at \(n = 5000\)); the covariate-dependent estimand, whose \(X\)-measurable variance model cannot represent that component, is no longer competitive and is in fact less precise than the homoskedastic estimand.

% The size of the efficiency gain reflects the difficulty of the variance function each estimand must learn. The covariate-dependent gain is substantial at both larger sample sizes (about sixteen percent at \(n = 2000\) and twelve percent at \(n = 5000\)) because \(\sigma^2(x)\) is a low-dimensional object that is learned accurately; out-of-bag leaf selection for its variance regression (Section \ref{sec:dml}) is what prevents \(\hat{\tau}_{\mathrm{cov}}^*\) from degenerating in the general design at \(n = 500\), where its misspecified \(X\)-measurable weight targets a pseudo-estimand rather than \(\tau^*_{\mathrm{cov}}\) (Remark \ref{rmk:eif_scope}), biasing it and inflating its standard error through occasional over-shooting (\(0.794\) against \(0.364\) for the correctly specified \(\hat{\tau}_{\mathrm{gen}}^*\)). The general gain is more modest (about four percent at \(n = 5000\)) because estimating the full dose-and-covariate variance surface \(\sigma^2(d,x)\) is more demanding, so a larger share of the population efficiency gain of Theorem \ref{thm:optimal_weights} is spent on nuisance estimation. This finite-sample cost of learning the variance-dependent weights, anticipated in Section \ref{sec:est}, tempers but does not overturn the efficiency ordering: in each design the estimand matched to the variance structure remains the most precise.

\begin{table}[htbp]
    \centering
    \caption{Monte Carlo Simulation Results}
    \label{tbl:sim_results}
    \footnotesize
    \setlength{\tabcolsep}{6pt}
    \renewcommand{\arraystretch}{0.88}
    \begin{tabular}{llcccc}
        \toprule
        \textbf{Design} & \textbf{Estimator} & \textbf{Estimate} & \textbf{Std. Error} & \textbf{Bias} & \textbf{95\% Cov.} \\
        \midrule
        \multicolumn{6}{l}{\textit{Panel A: Moderate Sample ($n = 500$)}} \\
        \midrule
        \textbf{1. Homoskedastic}
                        & $\hat{\tau}_{\mathrm{homo}}^*$ & 0.975     & 0.123     & -0.013    & 0.922     \\
                        & $\hat{\tau}_{\mathrm{cov}}^*$ & 1.013     & 0.150     & +0.025    & 0.959     \\
                        & $\hat{\tau}_{\mathrm{gen}}^*$ & 1.021     & 0.159     & +0.033    & 0.960     \\
                        & $\hat{\psi}$ & 0.975     & 0.140     & -0.025    & 0.939     \\
                        & $\hat{\tau}_{\mathrm{ADE}}$ & 0.882     & 0.232     & -0.117    & 0.923     \\
        \textbf{2. Covariate-Dep.}
                        & $\hat{\tau}_{\mathrm{homo}}^*$ & 0.973     & 0.235     & -0.019    & 0.932     \\
                        & $\hat{\tau}_{\mathrm{cov}}^*$ & 0.828     & 0.217     & +0.017    & 0.931     \\
                        & $\hat{\tau}_{\mathrm{gen}}^*$ & 0.837     & 0.226     & +0.026    & 0.937     \\
                        & $\hat{\psi}$ & 0.966     & 0.293     & -0.033    & 0.928     \\
                        & $\hat{\tau}_{\mathrm{ADE}}$ & 0.885     & 0.350     & -0.115    & 0.940     \\
        \textbf{3. General Het.}
                        & $\hat{\tau}_{\mathrm{homo}}^*$ & 0.982     & 0.378     & +0.007    & 0.926     \\
                        & $\hat{\tau}_{\mathrm{cov}}^*$ & 0.907     & 0.794     & -0.082    & 0.973     \\
                        & $\hat{\tau}_{\mathrm{gen}}^*$ & 0.947     & 0.364     & +0.009    & 0.950     \\
                        & $\hat{\psi}$ & 0.980     & 0.442     & -0.020    & 0.941     \\
                        & $\hat{\tau}_{\mathrm{ADE}}$ & 0.887     & 0.513     & -0.113    & 0.950     \\
        \midrule
        \multicolumn{6}{l}{\textit{Panel B: Large Sample ($n = 2000$)}} \\
        \midrule
        \textbf{1. Homoskedastic}
                        & $\hat{\tau}_{\mathrm{homo}}^*$ & 0.981     & 0.059     & -0.007    & 0.931     \\
                        & $\hat{\tau}_{\mathrm{cov}}^*$ & 0.998     & 0.069     & +0.010    & 0.945     \\
                        & $\hat{\tau}_{\mathrm{gen}}^*$ & 0.993     & 0.069     & +0.005    & 0.950     \\
                        & $\hat{\psi}$ & 0.985     & 0.063     & -0.015    & 0.944     \\
                        & $\hat{\tau}_{\mathrm{ADE}}$ & 0.896     & 0.094     & -0.104    & 0.813     \\
        \textbf{2. Covariate-Dep.}
                        & $\hat{\tau}_{\mathrm{homo}}^*$ & 0.974     & 0.121     & -0.018    & 0.926     \\
                        & $\hat{\tau}_{\mathrm{cov}}^*$ & 0.803     & 0.102     & -0.007    & 0.933     \\
                        & $\hat{\tau}_{\mathrm{gen}}^*$ & 0.810     & 0.105     & -0.000    & 0.932     \\
                        & $\hat{\psi}$ & 0.975     & 0.138     & -0.025    & 0.932     \\
                        & $\hat{\tau}_{\mathrm{ADE}}$ & 0.895     & 0.146     & -0.105    & 0.895     \\
        \textbf{3. General Het.}
                        & $\hat{\tau}_{\mathrm{homo}}^*$ & 0.967     & 0.198     & -0.008    & 0.919     \\
                        & $\hat{\tau}_{\mathrm{cov}}^*$ & 0.924     & 0.222     & -0.065    & 0.972     \\
                        & $\hat{\tau}_{\mathrm{gen}}^*$ & 0.924     & 0.174     & -0.015    & 0.931     \\
                        & $\hat{\psi}$ & 0.979     & 0.217     & -0.021    & 0.931     \\
                        & $\hat{\tau}_{\mathrm{ADE}}$ & 0.884     & 0.218     & -0.116    & 0.913     \\
        \midrule
        \multicolumn{6}{l}{\textit{Panel C: Very Large Sample ($n = 5000$)}} \\
        \midrule
        \textbf{1. Homoskedastic}
                        & $\hat{\tau}_{\mathrm{homo}}^*$ & 0.982     & 0.038     & -0.006    & 0.925     \\
                        & $\hat{\tau}_{\mathrm{cov}}^*$ & 0.991     & 0.042     & +0.003    & 0.942     \\
                        & $\hat{\tau}_{\mathrm{gen}}^*$ & 0.987     & 0.043     & -0.001    & 0.934     \\
                        & $\hat{\psi}$ & 0.988     & 0.041     & -0.012    & 0.920     \\
                        & $\hat{\tau}_{\mathrm{ADE}}$ & 0.897     & 0.059     & -0.103    & 0.553     \\
        \textbf{2. Covariate-Dep.}
                        & $\hat{\tau}_{\mathrm{homo}}^*$ & 0.979     & 0.074     & -0.013    & 0.928     \\
                        & $\hat{\tau}_{\mathrm{cov}}^*$ & 0.803     & 0.065     & -0.007    & 0.940     \\
                        & $\hat{\tau}_{\mathrm{gen}}^*$ & 0.805     & 0.066     & -0.005    & 0.928     \\
                        & $\hat{\psi}$ & 0.984     & 0.082     & -0.016    & 0.928     \\
                        & $\hat{\tau}_{\mathrm{ADE}}$ & 0.901     & 0.092     & -0.099    & 0.795     \\
        \textbf{3. General Het.}
                        & $\hat{\tau}_{\mathrm{homo}}^*$ & 0.979     & 0.118     & +0.004    & 0.952     \\
                        & $\hat{\tau}_{\mathrm{cov}}^*$ & 0.944     & 0.141     & -0.044    & 0.961     \\
                        & $\hat{\tau}_{\mathrm{gen}}^*$ & 0.929     & 0.114     & -0.010    & 0.916     \\
                        & $\hat{\psi}$ & 0.989     & 0.131     & -0.011    & 0.948     \\
                        & $\hat{\tau}_{\mathrm{ADE}}$ & 0.902     & 0.124     & -0.098    & 0.900     \\
        \bottomrule
    \end{tabular}
\end{table}

\section{Application: The Income Effect on Labor Supply}\label{sec:app}

We illustrate our proposed optimal estimands by re-examining the effect of unearned income on labor supply, utilizing the dataset from a survey of Massachusetts lottery winners analyzed by \citet{imbens2001estimating} and \citet{hiranok2004propensity} (hereafter HI04). The objective is to estimate the Marginal Propensity to Earn (MPE) out of unearned income, which corresponds to the average derivative of labor earnings with respect to the lottery prize.

The dataset consists of 237 individuals who won the Megabucks lottery in the mid-1980s. The continuous treatment \(D_i\) is the annualized value of the lottery prize, and the outcome \(Y_i\) is the average labor earnings recorded by the Social Security Administration approximately six years after winning. The covariates \(X_i\) include demographic characteristics (age, gender, education) and six years of pre-lottery earnings.

While the lottery prize itself is randomly assigned, the study design introduced potential confounding due to substantial nonresponse (around 50\%). HI04 demonstrated that nonresponse was correlated with the prize amount: winners of larger prizes were less likely to participate in the survey. This induced selection bias in the observed sample. For instance, men and individuals with higher pre-lottery earnings tended to have won larger prizes among the respondents. Following the literature, we maintain the assumption that conditional on the extensive set of covariates \(X_i\), the treatment assignment is unconfounded.

HI04 addressed this setting by introducing the Generalized Propensity Score (GPS) methodology to estimate the entire dose-response curve, \(\mu(d) = \E[Y_i(d)]\). Their analysis revealed significant heterogeneity in the MPE, \(\mu'(d)\). They estimated that the MPE ranges from approximately \(-0.10\) for small prizes (\$10,000) to \(-0.02\) for larger prizes (\$100,000), suggesting that the negative income effect is substantially stronger at lower levels of unearned income.

We compare the five estimators of Section \ref{sec:simulation} for summarizing the treatment effect.\footnote{We use the data provided in the GitHub repository of \citet{imbens2024comparing}. We failed to reproduce the exact results of HI04, as the summary statistics of pre-lottery earnings are slightly different from those reported in HI04. Our replication results are provided in Appendix \ref{sec:appendc}.} Table \ref{tbl:app} presents the estimation results. The point estimates are negative across all five estimators, but they differ markedly in both magnitude and precision, reflecting the distinct estimands implied by the different weighting schemes and the differing nuisance estimation burdens at the modest sample size of \(N = 237\).

\begin{table}[htbp]
    \centering
    \caption{Estimates of the Marginal Propensity to Earn (MPE) out of Lottery Prizes}
    \label{tbl:app}
    \begin{tabular}{llccc}
        \hline
        Estimand                                 & Method & Estimate & Std. Error & 95\% CI            \\
        \hline
        ADE ($\tau_{\mathrm{ADE}}$)              & GPS    & -0.0651  & (0.1310)   & [-0.3263, +0.2079] \\
        Homoskedastic ($\tau_{\mathrm{homo}}^*$) & DML    & -0.0240  & (0.0109)   & [-0.0454, -0.0027] \\
        Within-stratum (Hines $\psi$)            & DML    & -0.0435  & (0.0246)   & [-0.0918, +0.0047] \\
        Covariate-Dep. ($\tau_{\mathrm{cov}}^*$) & DML    & -0.0006  & (0.0019)   & [-0.0043, +0.0030] \\
        General Het. ($\tau_{\mathrm{gen}}^*$)   & DML    & -0.0037  & (0.0029)   & [-0.0094, +0.0019] \\
        \hline
    \end{tabular}

    \vspace{0.5em}
    \begin{minipage}{0.92\textwidth}
        {\footnotesize \textit{Note:} All DML estimators use $5$-fold cross-fitting. The ADE uses $1{,}000$ bootstrap replications.}
    \end{minipage}
\end{table}

The baseline GPS-based ADE estimate (\(-0.0651\)) is large in magnitude, but its 95\% bootstrap confidence interval is too wide to rule out zero or sizable positive effects. Among the optimal balancing estimands, the homoskedastic estimator \(\hat{\tau}_{\mathrm{homo}}^*\) yields a statistically significant negative estimate of \(-0.0240\), in line with the negative income effect documented by \citet{imbens2001estimating} and HI04. The within-stratum least-squares estimand \(\hat\psi = -0.0435\) is larger in magnitude but less precise, with a 95\% interval that narrowly includes zero. By contrast, \(\hat\tau_{\mathrm{cov}}^*\) and \(\hat\tau_{\mathrm{gen}}^*\) are close to zero with smaller standard errors, and both fail to reject the null hypothesis of zero effect.

Proposition \ref{prp:class} and Theorem \ref{thm:optimal_weights} can be used to provide a transparent interpretation of the target population for \(\tau_{\mathrm{homo}}^*\). It is a weighted average derivative where the weights are proportional to the conditional variance of the treatment, \(\Var(D_i \mid X_i)\), prioritizing individuals for whom the prize amount is most variable given their covariates. To interpret this weighting in the context of the lottery data, we consider the structure of \(\Var(D_i \mid X_i)\). HI04 modeled the log of the prize as conditionally normal: \(\log(D_i) \mid X_i \sim \mathcal{N}(\beta'X_i, \sigma^2)\). Under this log-normal structure, the conditional variance of the prize in levels is heteroskedastic and proportional to \(\exp(2\beta'X_i)\). This implies the variance is higher for individuals whose characteristics are associated with larger expected prizes. Because each of the optimal balancing estimands and \(\psi\) summarizes the same within-stratum slopes \(\Cov(D_i,Y_i\mid X_i)/\Var(D_i\mid X_i)\), differing primarily in how it weights covariate strata, these implied weights make their target populations explicit; Table \ref{tbl:target} reports the resulting composition. The GPS-based ADE shares \(\psi\)'s target population but, as a density-based average derivative, aggregates \(\mu'(D_i,X_i)\) over the dose distribution rather than through the within-stratum slope, which is why it differs from \(\psi\). The variance weight \(\Var(D_i\mid X_i)\propto\exp(2\beta'X_i)\) is monotone in the predicted prize and correlates \(0.62\) with being male, so \(\tau_{\mathrm{homo}}^*\) tilts toward the large-prize winners. Its implied target population is \(74\%\) male, against \(58\%\) in the sample, somewhat older, and concentrated, with the top decile of weight holding two-thirds of the mass. The shift toward men is robust across weighting choices. The pre-lottery-earnings channel is model-dependent. Under the log-normal structure the variance rises with the predicted prize, which is larger for higher earners \((\mathrm{corr}=0.49)\), though the cross-fitted weight in Table \ref{tbl:target} leaves their weighted share little changed.

The precision-weighted estimands reweight along a different margin. Rather than intensifying the tilt toward men and large prizes, \(\tau_{\mathrm{cov}}^*\) and \(\tau_{\mathrm{gen}}^*\) shift mass toward strata in which earnings are measured precisely: the weight-implied outcome variance falls roughly three- to fourfold, from \(130\) to \(29\) and \(39\) in Table \ref{tbl:target}. Because the conditional outcome variance is itself higher for men and for high earners in this sample, which correlates \(0.41\) and \(0.75\) with the two characteristics, this inverse-variance reweighting pulls partly against the prize-variance tilt, returning the share of male toward the sample value and lowering weighted pre-lottery earnings. Their smaller standard errors are therefore an inverse-variance efficiency gain, though both remain statistically indistinguishable from zero.

The within-stratum estimand \(\psi\) and the ADE instead retain the full sample as their target population. The estimand \(\psi\) averages the same within-stratum slopes as \(\tau_{\mathrm{homo}}^*\) but weights each covariate stratum equally rather than by \(\Var(D_i\mid X_i)\), so the gap between \(\hat\psi = -0.0435\) and \(\hat\tau_{\mathrm{homo}}^* = -0.0240\) is attributable to this reweighting toward the high-prize-variance strata, not to a different causal object. Placing equal weight on strata in which the slope is precisely and imprecisely identified is why \(\psi\) is less precise than \(\tau_{\mathrm{homo}}^*\), while the GPS-based ADE, despite sharing the same target, is the least precise of all because its representer divides by the estimated treatment density. \citet{hines2023optimally} caution that a weighting that depends on the conditional outcome variance can be difficult to interpret. However, the implied-weight analysis of Table \ref{tbl:target} addresses this concern by making each estimand's target population explicit through the weights of Proposition \ref{prp:class}. For this application, the homoskedastic estimand \(\tau_{\mathrm{homo}}^*\) is the interpretable and stable summary, and it remains informative where the other estimands are not.

\begin{table}[htbp]
    \centering
    \caption{Implied Target Population of the MPE Estimands (Lottery Application, $N=237$)}
    \label{tbl:target}
    {\footnotesize
        \setlength{\tabcolsep}{4pt}
        \begin{tabular}{lcccc}
            \hline
                                               & Sample                        & \multicolumn{3}{c}{Optimal balancing estimands}                                                             \\
            \cline{3-5}
                                               & $\tau_{\mathrm{ADE}}$, $\psi$ & $\tau_{\mathrm{homo}}^*$                        & $\tau_{\mathrm{cov}}^*$     & $\tau_{\mathrm{gen}}^*$     \\
            Implied stratum weight $g(x)$      & $1$                           & $\Var(D\mid X)$                                 & $\Var(D\mid X)/\sigma^2(X)$ & $\omega(D,X)(D-e_\omega)^2$ \\
            \hline
            Share male                         & 0.58                          & 0.74                                            & 0.60                        & 0.63                        \\
            Avg.\ pre-lottery earnings (\$000) & 12.8                          & 11.3                                            & 3.9                         & 4.1                         \\
            Mean age                           & 46.9                          & 52.1                                            & 59.6                        & 58.9                        \\
            Weight-implied outcome variance    & 130                           & 123                                             & 29                          & 39                          \\
            \hline
        \end{tabular}}

    \vspace{0.5em}
    \begin{minipage}{0.95\textwidth}
        {\footnotesize \textit{Note:} The ADE shares $\psi$'s uniform target but is a density-based average derivative. Weighted covariate means use the cross-fitted smooth weight $\E[g\mid X]$. The weight-implied outcome variance uses weight $\omega(D_i,X_i)(D_i-e_\omega(X_i))^2 = \E[B_{\mathrm{aug}}\mid D_i,X_i]$.}
    \end{minipage}
\end{table}

\section{Conclusion}\label{sec:conclusion}

\citet{hines2023optimally} showed that treating a bounded outcome weight, rather than a derivative weight, as the analytical primitive characterizes a class of continuous balancing weights and their optimally efficient members without estimating the conditional treatment density. This paper supplies the two ingredients that framework lacked for use in observational economic data. On identification, the optimal estimands require only conditional mean independence, which permits selection on the unobserved variance of the outcome, and they remain valid at sharp boundaries and at interior treatment deserts that the strict-overlap and density-smoothness conditions of classical theory rule out. On estimation, we solve the general heteroskedastic problem that the prior literature derived but declined to estimate. We obtain the efficient influence function, whose debiased precision weight neutralizes the first-order bias of variance estimation, together with the higher-order Debiased Machine Learning theory that controls the non-bilinear remainder it introduces and restores \(\sqrt{n}\)-consistency.

For empirical practice, we suggest specifying the outcome weight as the primitive and reporting the implied target population it induces, so that each estimand's population is explicit. The homoskedastic optimum is the partially linear regression coefficient, which requires no variance modeling, is straightforward to estimate. The precision-weighted estimands buy further efficiency when the conditional outcome variance can be modeled reliably, but they shift the target population toward the strata in which the outcome is measured most precisely, a trade-off the implied weights should be used to make explicit. Because the framework accommodates boundary mass points and interior treatment deserts without ad hoc trimming, it is well suited to economic settings with sharp thresholds and limited overlap.

%%%%%%%%%%%%%%%%%%%%%%%%%%%%%%%%%%%%%%%%%%%%%%%%%
\clearpage
\begin{singlespace}
  \bibliographystyle{ecta}
  \bibliography{references.bib}
\end{singlespace}
%%%%%%%%%%%%%%%%%%%%%%%%%%%%%%%%%%%%%%%%%%%%%%%%%

%%%%%%%%%%%%%%%%%%%%%%%%%%%%%%%%%%%%%%%%%%%%%%%%%
%%%%% These commands start the appendix and change the Table & Figure numbering
\newpage
\appendix
\setcounter{table}{0}
\renewcommand{\tablename}{Appendix Table}
\renewcommand{\figurename}{Appendix Figure}
\renewcommand{\thetable}{A\arabic{table}}
\setcounter{figure}{0}
\renewcommand{\thefigure}{A\arabic{figure}}
%%%%%%%%%%%%%%%%%%%%%%%%%%%%%%%%%%%%%%%%%%%%%%%%%

\section{Appendix: Identification and Proofs}\label{sec:appenda}

\numberwithin{equation}{section}
\setcounter{equation}{0}

Note: For notational simplicity, we omit individual subscripts in this appendix, writing \(X\) instead of \(X_i\), \(Y\) instead of \(Y_i\), \(D\) instead of \(D_i\), etc.

\subsection{Baseline Causal Identification}\label{sec:appenda_id}

This appendix collects the classical identification conditions referenced in Section \ref{sec:iden}, which fix the statistical target \(\tau_w = \E[w(D_i,X_i)\mu'(D_i,X_i)]\). Section \ref{sec:id} shows that, for the optimal estimands, these conditions can be substantially relaxed.
\begin{assumption}
  \label{ass:iden}
  \begin{enumerate}
    \item[(a)] \(D_i = d\) implies \(Y_i = Y_i(d)\).
    \item[(b)] \(\E[Y_i(d) \mid D_i, X_i] = \E[Y_i(d) \mid X_i]\) for all \(d \in \mathcal{D}\).
    \item[(c)] For any point \((d_0, x_0)\) in the support of \(P\), \(f(d|x_0)\) is continuous in \(d\) and strictly positive in an open neighborhood of \(d_0\).
    \item[(d)] The conditional expected potential outcome \(m(d,x)=\mathbb{E}[Y_i(d)|X_i=x]\) is continuously differentiable with respect to \(d\).
    \item[(e)] For some \(p, q \ge 1\) with \(1/p + 1/q = 1\), \(\E[|w(D_i,X_i)|^p] < \infty\). There exists a constant \(\varepsilon > 0\) and a measurable function \(L(d,x)\) with \(\mathbb{E}[L(D_i,X_i)^q] < \infty\) such that \(|m(D_i+\nu, X_i) - m(D_i, X_i)| \le L(D_i,X_i)|\nu|\) almost surely for all \(0 < |\nu| < \varepsilon\).
  \end{enumerate}
\end{assumption}

Assumption \ref{ass:iden} adapts traditional causal requirements to the continuous setting. Assumption \ref{ass:iden}(a) is the SUTVA assumption, and Assumption \ref{ass:iden}(b) is the weak unconfoundedness condition, requiring only conditional mean independence rather than full distributional independence. Assumption \ref{ass:iden}(c) is a localized overlap condition ensuring \(\mu(d,x)\) is well-defined in an open neighborhood of any observed point; Assumption \ref{ass:iden}(d) is a differentiability requirement on \(m(\cdot,x)\); and Assumption \ref{ass:iden}(e) is a moment condition bounding the weight function. The relaxed conditions of Section \ref{sec:id} (Assumption \ref{ass:id_min}) weaken (d) to continuity and drop the local overlap condition (c) altogether.
\begin{proposition}
  \label{prp:iden}
  Under Assumption \ref{ass:iden}, the weighted average causal derivative is identified by the observable statistical functional
  \[
    \dot{\tau}_w = \E[w(D_i, X_i)\mu'(D_i, X_i)].
  \]
\end{proposition}

For the classical analytical construction of the Riesz representer discussed in Section \ref{sec:instability}, we also invoke the following regularity conditions of \citet{powell1989semiparametric}.
\begin{assumption}[Regularity Condition for \citet{powell1989semiparametric}]\label{ass:powell}
  \begin{enumerate}
    \item[a.] The weighted conditional density \(w(d, x)f(d|x)\) is differentiable with respect to \(d\).
    \item[b.] \(w(\underline{d}, x)f(\underline{d}|x) = w(\bar{d}, x)f(\bar{d}|x) = 0\) where \(\underline{d}\) and \(\bar{d}\) are the lower and upper bounds of the support of \(D\).
    \item[c.] \(f(d|x) = 0\) implies \(w(d, x) = 0\).
  \end{enumerate}
\end{assumption}

\subsection{Proofs}\label{sec:appenda_proofs}

\begin{proof}[Proof of Proposition~\ref{prp:iden}]
  We evaluate the causal target parameter by analyzing the limit of the expected incremental effect
  \[
    \dot{\tau}_w = \lim_{\nu\rightarrow 0} \mathbb{E}\left[w(D_i,X_i)\frac{Y_i(D_i+\nu)-Y_i(D_i)}{\nu}\right].
  \]
  By the law of iterated expectations, conditioning on the observed treatment and covariates yields
  \[
    \dot{\tau}_w = \lim_{\nu\rightarrow 0} \mathbb{E}\left[w(D_i,X_i) \, \mathbb{E}\left[\frac{Y_i(D_i+\nu)-Y_i(D_i)}{\nu}\mathrel{\Bigg|}D_i,X_i\right]\right].
  \]
  By Assumptions \ref{ass:iden}(a) and \ref{ass:iden}(b), the conditional expectation of the potential outcome evaluated at any generic level \(d \in \mathcal{D}\) satisfies \(\mathbb{E}[Y_i(d)|D_i,X_i]=\mathbb{E}[Y_i(d)|X_i] \equiv m(d,X_i)\). Applying this equivalence at evaluation points \(d=D_i+\nu\) and \(d=D_i\), the inner conditional expectation simplifies to the difference quotient of the conditional expected potential outcome
  \[
    \mathbb{E}\left[\frac{Y_i(D_i+\nu)-Y_i(D_i)}{\nu}\mathrel{\Bigg|}D_i,X_i\right] = \frac{m(D_i+\nu,X_i)-m(D_i,X_i)}{\nu}.
  \]
  Thus, the estimand becomes
  \[
    \dot{\tau}_w = \lim_{\nu\rightarrow 0} \mathbb{E}\left[w(D_i,X_i)\frac{m(D_i+\nu,X_i)-m(D_i,X_i)}{\nu}\right].
  \]
  To exchange the limit and the expectation, we invoke the Dominated Convergence Theorem (DCT). Under the local Lipschitz bounding condition in Assumption \ref{ass:iden}(e), the argument inside the expectation is bounded by
  \[
    \left|w(D_i,X_i)\frac{m(D_i+\nu,X_i)-m(D_i,X_i)}{\nu}\right| \le |w(D_i,X_i)| \, L(D_i, X_i).
  \]
  We must verify that this bounding random variable is integrable. Applying H\"older's inequality yields
  \[
    \mathbb{E}\Big[|w(D_i,X_i)| \, L(D_i, X_i)\Big] \le \left(\mathbb{E}\Big[|w(D_i,X_i)|^p\Big]\right)^{1/p} \left(\mathbb{E}\Big[L(D_i, X_i)^q\Big]\right)^{1/q}.
  \]
  Assumption \ref{ass:iden}(e) guarantees that both \(\|w\|_p\) and \(\|L\|_q\) are finite, so that the bounding random variable is integrable (\(\mathbb{E}[|w|L] < \infty\)). Because the sequence of random variables is dominated by an integrable function that does not depend on \(\nu\), the DCT applies, allowing us to pass the limit inside the expectation
  \begin{align*}
    \dot{\tau}_w & = \mathbb{E}\left[w(D_i,X_i) \lim_{\nu\rightarrow 0}\frac{m(D_i+\nu,X_i)-m(D_i,X_i)}{\nu}\right] \\
                 & = \mathbb{E}\Big[w(D_i,X_i) \partial_d m(D_i,X_i)\Big],
  \end{align*}
  where the second equality follows from the continuous differentiability of \(m(\cdot,x)\) established in Assumption \ref{ass:iden}(d).

  Finally, we connect the structural derivative \(\partial_d m\) to the observable regression derivative \(\mu'\). For any observed point \((d_0,x_0)\) in the support of \(P\), Assumption \ref{ass:iden}(c) guarantees \(f(d|x_0)>0\) within an open neighborhood \(\mathcal{N}(d_0)\). For any \(d\in\mathcal{N}(d_0)\), the observable regression function satisfies
  \begin{align*}
    \mu(d,x_0) & = \mathbb{E}[Y_i|D_i=d,X_i=x_0]    \\
               & = \mathbb{E}[Y_i(d)|D_i=d,X_i=x_0] \\
               & = \mathbb{E}[Y_i(d)|X_i=x_0]       \\
               & = m(d,x_0).
  \end{align*}
  Because the observable function \(\mu(d,x)\) is identical to \(m(d,x)\) over an open neighborhood, their partial derivatives are also equal. Substituting this observable equivalence yields
  \[
    \dot{\tau}_w = \mathbb{E}\Big[w(D_i,X_i) \mu'(D_i,X_i)\Big].
  \]
\end{proof}

\begin{proof}[Proof of Proposition~\ref{prp:class}]
  We show two results. First, the estimand defined by a continuous balancing weight \(\alpha\) (Definition~\ref{def:class}) corresponds to a weighted average derivative with the specified weight \(w(d,x)\). Second, this weight is normalized, i.e., \(\E[w(D, X)] = 1\).

  Relating \(\alpha\) and \(w\). We need to verify that for any sufficiently smooth and bounded function \(\mu \in \mathcal{H}\), the following identity holds
  \begin{equation}
    \label{eq:proof_p2_identity}
    \E[\alpha(D, X)\mu(D, X)] = \E[w(D, X)\mu'(D, X)].
  \end{equation}
  We analyze this identity conditionally on \(X=x\). Let \(w(d,x)\) be defined as in Equation~\eqref{eq:weight}
  \[
    w(d, x) = -\frac{F(d|x)}{f(d|x)} \cdot \E[\alpha(D, X) \mid D \leq d, X = x].
  \]
  Let \(H(d, x) = w(d, x)f(d|x)\). We can rewrite \(H(d, x)\) using the definition of conditional expectation
  \begin{align*}
    H(d, x) & = - F(d|x) \cdot \E[\alpha(D, X) \mid D \leq d, X = x]       \\
            & = - \E[\alpha(D, X) \mathbb{I}(D \leq d) \mid X = x]         \\
            & = - \int_{\underline{d}}^{d} \alpha(t, x)f(t|x) \mathrm{d}t.
  \end{align*}

  We evaluate the right-hand side of Equation~\eqref{eq:proof_p2_identity} conditional on \(X=x\) using integration by parts
  \begin{align*}
    \E[w(D, X)\mu'(D, X) \mid X = x] & = \int_{\underline{d}}^{\bar{d}} w(d, x) \mu'(d, x) f(d|x) \mathrm{d}d                                                                     \\
                                     & = \int_{\underline{d}}^{\bar{d}} \mu'(d, x) H(d, x) \mathrm{d}d                                                                            \\
                                     & = [\mu(d, x)H(d, x)]_{\underline{d}}^{\bar{d}} - \int_{\underline{d}}^{\bar{d}} \mu(d, x) \frac{\partial H(d, x)}{\partial d} \mathrm{d}d.
  \end{align*}

  We must verify that the boundary terms vanish. At the lower bound \(\underline{d}\)
  \[
    H(\underline{d}, x) = - \int_{\underline{d}}^{\underline{d}} \alpha(t, x)f(t|x) \mathrm{d}t = 0.
  \]
  At the upper bound \(\bar{d}\)
  \[
    H(\bar{d}, x) = - \int_{\underline{d}}^{\bar{d}} \alpha(t, x)f(t|x) \mathrm{d}t = - \E[\alpha(D, X) \mid X = x].
  \]
  Since \(\alpha\) is a continuous balancing weight, \(\E[\alpha(D, X) \mid X = x] = 0\). Thus, \(H(\bar{d}, x) = 0\). The boundary terms vanish, since \(\mu(\cdot, x)\) has finite limits at the support endpoints.\footnote{If the support is unbounded (e.g., \(\mathbb{R}\)), we assume standard regularity conditions such that \(\lim_{d\to\pm\infty} \mu(d,x)H(d,x) = 0\).}

  Now we analyze the derivative of \(H(d, x)\). By the Fundamental Theorem of Calculus
  \[
    \frac{\partial H(d, x)}{\partial d} = \frac{\partial}{\partial d} \left( - \int_{\underline{d}}^{d} \alpha(t, x)f(t|x) \mathrm{d}t \right) = - \alpha(d, x)f(d|x).
  \]

  Substituting this back into the integration by parts formula
  \begin{align*}
    \E[w(D, X)\mu'(D, X) \mid X = x] & = 0 - \int_{\underline{d}}^{\bar{d}} \mu(d, x) [-\alpha(d, x)f(d|x)] \mathrm{d}d \\
                                     & = \int_{\underline{d}}^{\bar{d}} \mu(d, x) \alpha(d, x)f(d|x) \mathrm{d}d        \\
                                     & = \E[\alpha(D, X)\mu(D, X) \mid X = x].
  \end{align*}
  This confirms the identity in Equation~\eqref{eq:proof_p2_identity}.

  Normalization of \(w(d,x)\). We must show that \(\E[w(D, X)] = 1\). We utilize the identity established in Part 1, which holds for any smooth function \(\mu(d,x)\). We choose the test function \(\mu(d, x) = d\). Then \(\mu'(d, x) = 1\).

  Substituting this into the identity \(\E[w\mu' \mid X] = \E[\alpha\mu \mid X]\)
  \begin{align*}
    \E[w(D, X) \cdot 1 \mid X] & = \E[\alpha(D, X) \cdot D \mid X].
  \end{align*}
  Taking the expectation over \(X\) using the Law of Iterated Expectations
  \begin{align*}
    \E[w(D, X)] & = \E[\E[w(D, X) \mid X]]       \\
                & = \E[\E[\alpha(D, X)D \mid X]] \\
                & = \E[\alpha(D, X)D].
  \end{align*}
  Since \(\alpha\) is a continuous balancing weight, the normalization constraint ensures \(\E[\alpha(D, X)D] = 1\). Therefore, \(\E[w(D, X)] = 1\).
\end{proof}

\begin{proof}[Proof of Lemma~\ref{lem:nonneg}]
  From Equation~\eqref{eq:weight}, since \(F(d|x) \geq 0\) and \(f(d|x) \geq 0\) for all \(d, x\), it is straightforward to show that \(w(d, x) \geq 0\) if and only if
  \[
    \E[\alpha(D, X) \mid D \leq d, X = x] \leq 0 \ \text{for all } d, x.
  \]
\end{proof}

\begin{proof}[Proof of Proposition~\ref{prp:id_min}]
  We work conditionally on \(X = x\) and suppress \(x\) where convenient. By consistency (Assumption \ref{ass:id_min}(a)) and mean independence (Assumption \ref{ass:id_min}(b)), for \(P\)-almost every observed \((d,x)\),
  \[
    \mu(d,x) = \E[Y \mid D=d, X=x] = \E[Y(d) \mid D=d, X=x] = \E[Y(d) \mid X=x] = m(d,x).
  \]
  Hence \(\tau_\alpha = \E[\alpha(D,X) Y] = \E[\alpha(D,X)\mu(D,X)] = \E[\alpha(D,X) m(D,X)]\), where the outer expectation is over the observed law of \((D,X)\); this step uses no positivity of \(f(d|x)\).

  Fix \(x\) and write the inner conditional expectation as a Lebesgue--Stieltjes integral against the conditional distribution \(F(\cdot|x)\) of \(D\) given \(X=x\), supported on \([\underline{d}, \overline{d}]\),
  \[
    \E[\alpha(D,X) m(D,X) \mid X = x] = \int_{[\underline{d}, \overline{d}]} \alpha(d,x) m(d,x)\, \mathrm{d}F(d|x).
  \]
  By definition, \(W_\alpha(d,x) = -\E[\alpha(D,X)\mathbb{I}(D \le d) \mid X=x] = -\int_{[\underline{d}, d]} \alpha(t,x)\, \mathrm{d}F(t|x)\), which is precisely the cumulative object \(H(d,x)\) constructed in the proof of Proposition \ref{prp:class}. It is right-continuous and of bounded variation, with total variation \(\E[|\alpha(D,X)| \mid X=x] < \infty\) (as \(\alpha(\cdot,x) \in L_1(P_{D|X})\) for almost every \(x\)), and satisfies \(\mathrm{d}(-W_\alpha)(d,x) = \alpha(d,x)\,\mathrm{d}F(d|x)\) as signed measures. The proof of Proposition \ref{prp:class} carries out the corresponding integration by parts for the observed regression \(\mu\) under a density; we reprise it here for \(m\) in Lebesgue--Stieltjes form, which additionally accommodates atoms in the treatment distribution. Therefore
  \[
    \int_{[\underline{d}, \overline{d}]} \alpha(d,x) m(d,x)\, \mathrm{d}F(d|x) = \int_{[\underline{d}, \overline{d}]} m(d,x)\, \mathrm{d}(-W_\alpha)(d,x).
  \]
  We apply integration by parts for Lebesgue--Stieltjes integrals. Because \(m(\cdot,x)\) is absolutely continuous, hence continuous, by Assumption \ref{ass:id_min}(c), and \(W_\alpha(\cdot,x)\) is of bounded variation with jumps only at atoms of \(D\), the two share no common points of discontinuity, so
  \[
    \int_{[\underline{d}, \overline{d}]} m\, \mathrm{d}(-W_\alpha) = \big[-m(d,x) W_\alpha(d,x)\big]_{\underline{d}}^{\overline{d}} + \int_{[\underline{d}, \overline{d}]} W_\alpha(d,x)\, \mathrm{d}m(d,x).
  \]
  The boundary term vanishes: at the lower endpoint \(W_\alpha(\underline{d}^-, x) = 0\), and at the upper endpoint \(W_\alpha(\overline{d}, x) = -\E[\alpha(D,X) \mid X=x] = 0\) by the covariate balance constraint of Definition \ref{def:class}. The endpoint values of \(m\) are finite by Assumption \ref{ass:id_min}(c), so \([-m W_\alpha]_{\underline{d}}^{\overline{d}} = 0\). Because \(m(\cdot,x)\) is absolutely continuous by Assumption \ref{ass:id_min}(c), \(\mathrm{d}m(d,x) = \partial_d m(d,x)\,\mathrm{d}d\), giving
  \[
    \E[\alpha(D,X) m(D,X) \mid X = x] = \int_{\underline{d}(x)}^{\overline{d}(x)} \frac{\partial m(d,x)}{\partial d}\, W_\alpha(d,x)\, \mathrm{d}d.
  \]
  Taking the expectation over \(X\) and interchanging it with the inner integral---justified by Fubini's theorem, since the integrand \(\partial_d m(d,X)\, W_\alpha(d,X)\) is dominated by \(L(d,X)|W_\alpha(d,X)|\), which is integrable by the dominated-derivative condition Assumption \ref{ass:id_min}(d)---yields Equation \eqref{eq:id_repr}.

  Finally, for the optimal weights \(\alpha^*_\omega\) of Theorem \ref{thm:optimal_weights}, the implied weight satisfies \(W_{\alpha^*_\omega}(d,x) = -F(d|x)\,\E[\alpha^*_\omega(D,X) \mid D \le d, X=x] \ge 0\), since the truncated conditional expectation is non-positive by Lemma \ref{lem:nonneg} (verified for all three optimal estimands in Remark 2). These weights exist and lie in \(L_2(P_{D,X})\) under Assumption \ref{ass:regularity_eif}, so they satisfy the hypotheses of the proposition.
\end{proof}

\begin{proof}[Proof of Theorem \ref{thm:optimal_weights}]
  We seek the optimal continuous balancing weight \(\alpha_\omega\) that minimizes the variance component \(V_S = \E[\alpha_\omega^2(D_i,X_i)\sigma^2(D_i,X_i)]\). Let \(\omega(D_i,X_i) = \sigma^{-2}(D_i, X_i)\) denote the true precision weight. The objective simplifies to minimizing \(\E[\alpha_\omega^2(D_i,X_i)\omega^{-1}(D_i,X_i)]\). Let \(\mathcal{H} = L_2(P_{D,X})\). Under Assumption \ref{ass:regularity_eif}, \(\omega(D_i,X_i)\) is strictly positive and bounded, ensuring the objective is well-defined. We formulate this as a constrained functional optimization problem in the Hilbert space \(\mathcal{H}\),
  \begin{align*}
    \min_{\alpha \in \mathcal{H}} \quad & \E\big[\alpha^2(D,X) \omega^{-1}(D,X)\big]             \\
    \text{subject to} \quad             & (C1):\E[\alpha(D,X)D] = 1,                             \\
                                        & (C2):\E[\alpha(D,X) \mid X] = 0 \text{ almost surely.}
  \end{align*}
  We employ the method of Lagrange multipliers. Let \(\lambda \in \mathbb{R}\) be the multiplier for constraint (C1), and let \(\eta(X) \in L_2(P_X)\) be the functional multiplier for (C2). The Lagrangian is
  \[
    L(\alpha, \lambda, \eta) = \E\big[\alpha^2(D,X)\omega^{-1}(D,X) - 2\lambda\alpha(D,X)D - 2\eta(X)\alpha(D,X)\big] + 2\lambda.
  \]
  To find the optimum, we compute the G\^{a}teaux derivative of \(L\) with respect to \(\alpha\) in an arbitrary direction \(h \in \mathcal{H}\) and set it to zero
  \begin{align*}
    \nabla_\alpha L(\alpha; h) & = \lim_{\epsilon \rightarrow 0} \frac{L(\alpha + \epsilon h, \lambda, \eta) - L(\alpha, \lambda, \eta)}{\epsilon}                                              \\
                               & = \frac{d}{d\epsilon} \E\big[(\alpha + \epsilon h)^2 \omega^{-1}(D,X) - 2\lambda(\alpha + \epsilon h)D - 2\eta(X)(\alpha + \epsilon h)\big] \Big|_{\epsilon=0} \\
                               & = 2\E\big[h(D,X)\big(\alpha(D,X)\omega^{-1}(D,X) - \lambda D - \eta(X)\big)\big] = 0.
  \end{align*}
  Because the derivative must equal zero for all arbitrary test functions \(h \in \mathcal{H}\), the term inside the parenthesis must be zero almost surely. This establishes the first-order condition
  \[
    \alpha(D,X) = \omega(D,X)\big(\lambda D + \eta(X)\big).
  \]

  We apply the continuous covariate balance constraint (C2), \(\E[\alpha(D,X) \mid X] = 0\), to solve for \(\eta(X)\)
  \[
    \E\big[\omega(D,X)(\lambda D + \eta(X)) \mid X \big] = 0 \implies \lambda \E[\omega(D,X)D \mid X] + \eta(X)\E[\omega(D,X) \mid X] = 0.
  \]
  Solving for \(\eta(X)\) directly yields the precision-weighted centering term
  \[
    \eta(X) = -\lambda \frac{\E[\omega(D,X)D \mid X]}{\E[\omega(D,X) \mid X]} \equiv -\lambda e_\omega(X).
  \]
  Substituting \(\eta(X)\) back into the expression for \(\alpha(D,X)\) provides the unnormalized optimal weight
  \[
    \alpha(D,X) = \lambda \omega(D,X)\big(D - e_\omega(X)\big).
  \]

  Finally, we apply the normalization constraint (C1), \(\E[\alpha(D,X)D] = 1\), to uniquely identify the scalar multiplier \(\lambda\)
  \[
    1 = \lambda \E\big[ \omega(D,X)\big(D - e_\omega(X)\big)D \big].
  \]
  We decompose \(D\) as \(D = (D - e_\omega(X)) + e_\omega(X)\), separating the expectation
  \[
    1 = \lambda \E\big[ \omega(D,X)\big(D - e_\omega(X)\big)^2 \big] + \lambda \E\big[ \omega(D,X)\big(D - e_\omega(X)\big)e_\omega(X) \big].
  \]
  By the Law of Iterated Expectations, the rightmost term evaluates exactly to zero because
  \begin{align*}
      & \E\big[ \omega(D,X)(D - e_\omega(X))e_\omega(X) \big]                                               \\
    = & \E_X\Big[ e_\omega(X) \E[ \omega(D,X)D - \omega(D,X)e_\omega(X) \mid X ] \Big]                      \\
    = & \E_X\Big[ e_\omega(X) \big( \E[\omega(D,X)D \mid X] - e_\omega(X)\E[\omega(D,X) \mid X] \big) \Big] \\
    = & 0,
  \end{align*}
  which follows from the definition of \(e_\omega(X)\). This leaves the scalar multiplier
  \[
    \lambda = \left( \E\big[ \omega(D,X)\big(D - e_\omega(X)\big)^2 \big] \right)^{-1}.
  \]
  Substituting \(\lambda\) completes the derivation of the unified formula for \(\alpha^*_\omega(d,x)\). The optimal estimand \(\tau^*_\omega = \E[\alpha^*_\omega(D, X) Y]\) follows trivially. The specific forms in cases (a), (b), and (c) arise by defining \(\omega(d,x)\) as \(\sigma^{-2}(d,x)\), \(\sigma^{-2}(x)\), and \(1\), respectively.
\end{proof}

\begin{proof}[Proof of Theorem \ref{thm:eif_optimal}]
  We use the method of pathwise differentiation. Let \(P_\epsilon\) be a regular one-dimensional parametric submodel passing through the true distribution \(P_0\) at \(\epsilon=0\), with score function \(S(O) = \frac{\partial}{\partial \epsilon} \log p_\epsilon(O) \big|_{\epsilon=0}\). Assumption \ref{ass:regularity_eif} ensures all required moments are finite.

  The estimand is \(\tau^*_\omega = N_\omega / K_\omega\), where the numerator is
  \[
    N_\omega = \E[\omega(D, X)(D - e_\omega(X))(\mu(D,X) - \rho_\omega(X))]
  \]
  and the denominator is
  \[
    K_\omega = \E[\omega(D, X)(D - e_\omega(X))^2]
  \]
  By the functional Delta method, the EIF of the ratio is \(\phi^*(O) = \frac{1}{K_\omega}(\phi_N(O) - \tau^*_\omega \phi_K(O))\), where \(\phi_N\) and \(\phi_K\) are the EIFs of the numerator and denominator.

  We calculate the pathwise derivative of the numerator \(N_\omega\) evaluated at \(\epsilon=0\), denoted with a tilde
  \begin{align}
    \tilde{N}_\omega & = \E\big[ \tilde{\omega}(D - e_\omega)(\mu - \rho_\omega) - \omega \tilde{e}_\omega (\mu - \rho_\omega) - \omega (D - e_\omega) \tilde{\rho}_\omega + \omega(D - e_\omega)\tilde{\mu} \nonumber \\
                     & \quad + \omega(D - e_\omega)(\mu - \rho_\omega)S(O) \big]. \label{eq:pathwise_num}
  \end{align}
  Because \(e_\omega\) and \(\rho_\omega\) are defined via orthogonal projections, \(\E[\omega(D - e_\omega) \mid X] = 0\) and \(\E[\omega(\mu - \rho_\omega) \mid X] = 0\). By the law of iterated expectations, any cross-terms involving \(\tilde{e}_\omega\) or \(\tilde{\rho}_\omega\) vanish
  \[
    \E[\omega \tilde{e}_\omega (\mu - \rho_\omega)] = \E_X\big[ \tilde{e}_\omega(X) \E[\omega(\mu - \rho_\omega) \mid X] \big] = 0, \quad \text{and} \quad \E[\omega (D - e_\omega) \tilde{\rho}_\omega] = 0.
  \]
  The pathwise derivative of the conditional mean is \(\tilde{\mu}(D,X) = \E[(Y - \mu)S(O) \mid D,X]\). Thus,
  \[
    \E[\omega(D - e_\omega)\tilde{\mu}] = \E\big[\omega(D - e_\omega)(Y - \mu)S(O)\big].
  \]

  We now evaluate the pathwise derivative of the precision weight \(\tilde{\omega}\). Since \(\omega(D,X) = 1/\sigma^2(D,X)\) where \(\sigma^2(D,X) = \E[(Y - \mu)^2 \mid D,X]\), the chain rule yields
  \[
    \tilde{\omega} = -\frac{1}{\sigma^4}\tilde{\sigma}^2 = -\omega^2 \E\big[((Y - \mu)^2 - \sigma^2)S(O) \mid D,X\big].
  \]
  Substituting this into the first term of Equation \eqref{eq:pathwise_num}
  \[
    \E\big[\tilde{\omega}(D - e_\omega)(\mu - \rho_\omega)\big] = -\E\big[\omega^2((Y - \mu)^2 - \sigma^2)(D - e_\omega)(\mu - \rho_\omega)S(O)\big].
  \]
  Collecting all non-zero terms in \(\tilde{N}_\omega\), we identify the influence function for the numerator, \(\phi_N(O)\), such that \(\tilde{N}_\omega = \E[\phi_N(O)S(O)]\)
  \begin{align*}
    \phi_N(O) & = \big[\omega - \omega^2((Y - \mu)^2 - \sigma^2)\big](D - e_\omega)(\mu - \rho_\omega) + \omega(D - e_\omega)(Y - \mu) - N_\omega \\
              & = \big[2\omega - \omega^2(Y - \mu)^2\big](D - e_\omega)(\mu - \rho_\omega) + \omega(D - e_\omega)(Y - \mu) - N_\omega             \\
              & \equiv A_{\mathrm{aug}}(O; \eta_0) - N_\omega,
  \end{align*}
  where we define \(\omega_{\mathrm{db}}(O) = 2\omega - \omega^2(Y - \mu)^2\).

  An identical pathwise differentiation applied to the denominator \(K_\omega\) yields its influence function. The projection terms identically vanish, leaving
  \[
    \phi_K(O) = \big[\omega - \omega^2((Y - \mu)^2 - \sigma^2)\big](D - e_\omega)^2 - K_\omega = \omega_{\mathrm{db}}(D - e_\omega)^2 - K_\omega \equiv B_{\mathrm{aug}}(O; \eta_0) - K_\omega.
  \]

  Finally, combining the components yields the overall unscaled EIF as defined in \(\psi_{\mathrm{aug}}\).
\end{proof}

\begin{proof}[Proof of Theorem \ref{thm:dml_asymptotics}]
  Let \(P\) denote the true probability measure. To account for cross-fitting, let \(\mathcal{I}_k\) denote the observation indices in fold \(k\), let \(P_{n,k}\) denote the empirical measure over fold \(k\) (i.e., \(P_{n,k}[g] = |\mathcal{I}_k|^{-1} \sum_{i \in \mathcal{I}_k} g(O_i)\)), and let \(P_n = \sum_{k=1}^K \frac{|\mathcal{I}_k|}{n} P_{n,k}\) denote the full empirical measure. Let \(\hat{\eta}_k = (\hat{\omega}_k, \hat{\mu}_k, \hat{e}_k, \hat{\rho}_k)\) denote the nuisance estimators trained on the complement \(\mathcal{I}_k^c\). For notational simplicity, we drop the \(\omega\) subscripts on \(e\) and \(\rho\), denote the true parameters as \(\eta_0\), use \(\|\cdot\|_{P,2}\) for the \(L_2(P)\) norm, and use \(\|\cdot\|_\infty\) for the uniform supremum norm. By Assumption \ref{ass:regularity_eif}(a) and algorithmic trimming, the precision weights \(\omega_0\) and \(\hat{\omega}_k\) are bounded such that \(\|\omega_0\|_\infty \le C\), \(\|1/\omega_0\|_\infty \le C\), and \(\|\hat{\omega}_k\|_\infty \le C\).

  We define the moment functions for the numerator and denominator, corresponding to the components of the augmented score \(\psi_{\mathrm{aug}}\), as
  \begin{align*}
    m_A(O; \eta) & = \omega_{\mathrm{db}}(O; \eta)(D - e(X))(\mu(D,X) - \rho(X)) + \omega(D,X)(D - e(X))(Y - \mu(D,X)), \\
    m_B(O; \eta) & = \omega_{\mathrm{db}}(O; \eta)(D - e(X))^2,
  \end{align*}
  where \(\omega_{\mathrm{db}}(O; \eta) = 2\omega(D,X) - \omega^2(D,X)(Y - \mu(D,X))^2\).

  The population parameters are \(A_0 = P[m_A(O; \eta_0)]\) and \(B_0 = P[m_B(O; \eta_0)]\). Note \(B_0 = K_{\omega_0}\) and the optimal estimand is \(\tau_{\omega}^* = A_0/B_0\). The cross-fitted estimators are \(\hat{A} = P_n[m_A(O; \hat{\eta}_k)]\) and \(\hat{B} = P_n[m_B(O; \hat{\eta}_k)]\).

  The true influence functions are \(\phi_A(O) = m_A(O; \eta_0) - A_0\) and \(\phi_B(O) = m_B(O; \eta_0) - B_0\). We decompose the estimation error for the denominator \(\hat{B}\) as
  \begin{align*}
    \hat{B} - B_0 & = \sum_{k=1}^K \frac{|\mathcal{I}_k|}{n} (P_{n,k} - P)[m_B(O; \hat{\eta}_k) - m_B(O; \eta_0)]            \\
                  & \quad + \sum_{k=1}^K \frac{|\mathcal{I}_k|}{n} P[m_B(O; \hat{\eta}_k) - m_B(O; \eta_0)] + P_n[\phi_B(O)] \\
                  & \equiv E_n^B + R_n^B + P_n[\phi_B(O)],
  \end{align*}
  where \(E_n^B\) is the empirical process term and \(R_n^B\) is the remainder bias term. An analogous decomposition \(\hat{A} - A_0 = E_n^A + R_n^A + P_n[\phi_A(O)]\) holds for the numerator. We must show that the remainders and empirical process terms are \(o_P(n^{-1/2})\) under either Condition A or Condition B.

  The Remainder Term \(R_n^B\). We analyze the fold-specific remainder \(R_{n,k}^B = P[m_B(O; \hat{\eta}_k) - m_B(O; \eta_0)]\) using the Law of Iterated Expectations conditional on \((D,X)\). Let \(\Delta\omega = \hat{\omega}_k - \omega_0\), \(\Delta\mu = \hat{\mu}_k - \mu_0\), \(\Delta e = \hat{e}_k - e_0\), and \(\Delta\rho = \hat{\rho}_k - \rho_0\).

  Because \(\E[(Y - \hat{\mu}_k)^2 \mid D,X] = \sigma_0^2 + (\mu_0 - \hat{\mu}_k)^2 = 1/\omega_0 + \Delta\mu^2\), the conditional expectation of the debiased weight evaluates to
  \begin{align*}
    \E[\omega_{\mathrm{db}}(O; \hat{\eta}_k) \mid D,X] & = 2\hat{\omega}_k - \hat{\omega}_k^2\left(\frac{1}{\omega_0} + \Delta\mu^2\right)                                                 \\
                                                       & = 2(\omega_0 + \Delta\omega) - \frac{\omega_0^2 + 2\omega_0\Delta\omega + \Delta\omega^2}{\omega_0} - \hat{\omega}_k^2\Delta\mu^2 \\
                                                       & = \omega_0 - \left(\frac{\Delta\omega^2}{\omega_0} + \hat{\omega}_k^2\Delta\mu^2\right) \equiv \omega_0 - \epsilon_{\mathrm{db}},
  \end{align*}
  where \(\epsilon_{\mathrm{db}}\) captures the second-order estimation errors of the precision weight. Substituting this into the conditional expectation of \(m_B(O; \hat{\eta}_k)\) yields
  \[
    \E[m_B(O; \hat{\eta}_k) \mid D,X] = \omega_0(D - e_0)^2 - 2\omega_0(D - e_0)\Delta e + \omega_0\Delta e^2 - \epsilon_{\mathrm{db}}(D - \hat{e}_k)^2.
  \]
  The leading term recovers \(\E[m_B(O; \eta_0) \mid D,X]\). Taking the unconditional expectation, the first-order error vanishes because \(P[-2\omega_0(D - e_0)\Delta e] = -2P[\Delta e \E[\omega_0(D - e_0) \mid X]] = 0\), which confirms Neyman orthogonality. The remainder simplifies to
  \[
    R_{n,k}^B = P[\omega_0\Delta e^2 - \epsilon_{\mathrm{db}}(D - \hat{e}_k)^2].
  \]

  Under Condition A, because all variables and estimators are bounded (\(\|\cdot\|_\infty \le C\)), applying the Cauchy-Schwarz inequality yields \(|R_{n,k}^B| \le C\|\Delta e\|_{P,2}^2 + C\|\Delta\omega\|_{P,2}^2 + C\|\Delta\mu\|_{P,2}^2 = o_P(n^{-1/2})\).

  Under Condition B, the projection error satisfies \(|P[\omega_0\Delta e^2]| \le \|\omega_0\|_\infty \|\Delta e\|_{P,2}^2 = o_P(n^{-1/2})\). Notice that \((D - \hat{e}_k) = (D - e_0) - \Delta e\). Thus,
  \[
    (D - \hat{e}_k)^2 \le 2(D - e_0)^2 + 2\Delta e^2 \le 2v(O)^2 + 2\Delta e^2.
  \]
  Evaluating the \(\Delta\mu^2\) component of \(\epsilon_{\mathrm{db}}\)
  \begin{align*}
    |P[\hat{\omega}_k^2\Delta\mu^2(D - \hat{e}_k)^2]| & \le 2\|\hat{\omega}_k\|_\infty^2 \left( P[\Delta\mu^2 v(O)^2] + P[\Delta\mu^2 \Delta e^2] \right)                               \\
                                                      & \le 2\|\hat{\omega}_k\|_\infty^2 \left( \|\Delta\mu \cdot v(O)\|_{P,2}^2 + \|\Delta\mu\|_\infty^2 \|\Delta e\|_{P,2}^2 \right).
  \end{align*}
  Condition B(iii) guarantees \(\|\Delta\mu \cdot v(O)\|_{P,2}^2 = o_P(n^{-1/2})\). Because Condition B(ii) guarantees \(\|\Delta\mu\|_\infty = o_P(1)\), the second term is \(o_P(1) \cdot o_P(n^{-1/2}) = o_P(n^{-1/2})\). The \(\Delta\omega^2/\omega_0\) component follows identically. Thus, \(R_n^B = o_P(n^{-1/2})\).

  The Remainder Term \(R_n^A\). Conditioning on \((D,X)\) for the numerator \(m_A(O; \hat{\eta}_k)\) and applying the internal cancellation \(\hat{\mu}_k - \hat{\rho}_k - \Delta\mu = \mu_0 - \rho_0 - \Delta\rho\), the first-order errors again vanish by orthogonality: \(P[-\omega_0(D - e_0)\Delta\rho] = 0\) and \(P[-\omega_0\Delta e(\mu_0 - \rho_0)] = 0\). The remainder evaluates to
  \[
    R_{n,k}^A = P[\omega_0\Delta e\Delta\rho - \Delta\omega(D - \hat{e}_k)\Delta\mu - \epsilon_{\mathrm{db}}(D - \hat{e}_k)(\hat{\mu}_k - \hat{\rho}_k)].
  \]
  Under Condition A, utilizing uniform boundedness, the Cauchy-Schwarz inequality yields bounds of the form \(C\|\Delta e\|_{P,2}\|\Delta\rho\|_{P,2}\), which scale at \(o_P(n^{-1/2})\).

  Under Condition B, for the first term: \(|P[\omega_0\Delta e\Delta\rho]| \le \|\omega_0\|_\infty \|\Delta e\|_{P,2} \|\Delta\rho\|_{P,2} = o_P(n^{-1/2})\).

  For the cross-term, we decompose \((D-\hat{e}_k) = (D-e_0) - \Delta e\). By the triangle inequality, \(|D-\hat{e}_k| \le |D-e_0| + |\Delta e| \le v(O) + |\Delta e|\). Applying the triangle and Cauchy-Schwarz inequalities,
  \begin{align*}
    |P[\Delta\omega(D - \hat{e}_k)\Delta\mu]| & \le P[|\Delta\omega| \cdot v(O) \cdot |\Delta\mu|] + P[|\Delta\omega| \cdot |\Delta e| \cdot |\Delta\mu|]                   \\
                                              & \le \|\Delta\omega \cdot v(O)\|_{P,2} \|\Delta\mu\|_{P,2} + \|\Delta\omega\|_\infty \|\Delta e\|_{P,2} \|\Delta\mu\|_{P,2}.
  \end{align*}
  Because \(v(O) \ge 1\), the unweighted rate is bounded by \(\|\Delta\mu\|_{P,2} \le \|\Delta\mu \cdot v(O)\|_{P,2}\). Thus, the first term is \(o_P(n^{-1/4}) o_P(n^{-1/4}) = o_P(n^{-1/2})\). The second term is \(o_P(1) \cdot o_P(n^{-1/4}) \cdot o_P(n^{-1/4}) = o_P(n^{-1/2})\).

  For the \(\epsilon_{\mathrm{db}}\) components, we apply the triangle inequality to the product
  \[
    |(D-\hat{e}_k)(\hat{\mu}_k - \hat{\rho}_k)| \le (v(O) + |\Delta e|)(v(O) + |\Delta\mu| + |\Delta\rho|).
  \]
  We can extract the uniform estimation errors outside the expectation to rigorously bound the \(\Delta\mu^2\) component unconditionally. Define the uniform bound \(C_n = (1 + \|\Delta e\|_\infty)(1 + \|\Delta\mu\|_\infty + \|\Delta\rho\|_\infty) = \mathcal{O}_P(1)\). Because \(v(O) \ge 1\), the product is bounded by \(C_n v(O)^2\). Thus,
  \[
    |P[\hat{\omega}_k^2 \Delta\mu^2(D - \hat{e}_k)(\hat{\mu}_k - \hat{\rho}_k)]| \le \|\hat{\omega}_k\|_\infty^2 C_n P[\Delta\mu^2 v(O)^2] = \mathcal{O}_P(1) \|\Delta\mu \cdot v(O)\|_{P,2}^2 = o_P(n^{-1/2}).
  \]
  The \(\Delta\omega^2/\omega_0\) component follows identically. Thus, \(R_n^A = o_P(n^{-1/2})\).

  The Empirical Process Terms \(E_n^A\) and \(E_n^B\). Because the nuisance estimators \(\hat{\eta}_k\) are computed on the independent sample \(\mathcal{I}_k^c\), they act as conditionally centered empirical processes. By Lemma 2 of \citet{kennedy2020sharp}, these processes scale as \(\mathcal{O}_P(n^{-1/2} \|m_J(O; \hat{\eta}_k) - m_J(O; \eta_0)\|_{P,2})\) for \(J \in \{A, B\}\). We must verify that the \(L_2(P)\) norm of the estimated score difference is \(o_P(1)\).

  Under Condition A, because all variables and estimators are uniformly bounded, the mapping \(\eta \mapsto m_J(O; \eta)\) has bounded partial derivatives and is therefore globally Lipschitz continuous. Minkowski's inequality ensures \(\|m_J(O; \hat{\eta}_k) - m_J(O; \eta_0)\|_{P,2} \le L \big( \|\Delta\omega\|_{P,2} + \|\Delta\mu\|_{P,2} + \|\Delta e\|_{P,2} + \|\Delta\rho\|_{P,2} \big) = o_P(1)\), where \(L > 0\) is the Lipschitz constant.

  Under Condition B, the leading first-order discrepancy term in \(m_B(O; \eta)\) is proportional to \(\Delta\omega(Y - \mu_0)^2(D - e_0)^2\). We bound its \(L_2(P)\) norm as
  \begin{align*}
    \|\Delta\omega(Y - \mu_0)^2(D - e_0)^2\|_{P,2} & \le \|\Delta\omega\|_\infty \left( \E[(Y - \mu_0)^4(D - e_0)^4] \right)^{1/2}       \\
                                                   & \le \|\Delta\omega\|_\infty \left( \E[(Y - \mu_0)^8] \E[(D - e_0)^8] \right)^{1/4}.
  \end{align*}

  Because Condition B(i) guarantees finite moments up to \(q \ge 8\), and true conditional expectations inherit bounds from the data via Jensen's inequality (e.g., \(\E[\mu_0(X)^8] \le \E[Y^8]\)), these structural expectations are finite constants. Because the uniform estimation error converges as \(\|\Delta\omega\|_\infty = o_P(1)\), the entire product vanishes to \(o_P(1)\).

  Quadratic estimation error terms are handled identically
  \[
    \|\hat{\omega}_k^2 \Delta\mu^2 (D - e_0)^2\|_{P,2} \le \|\hat{\omega}_k\|_\infty^2 \|\Delta\mu\|_\infty^2 \left( \E[(D - e_0)^4] \right)^{1/2} = \mathcal{O}_P(1) \cdot o_P(1) \cdot \mathcal{O}(1) = o_P(1).
  \]
  Identical logic applies to all remaining lower-order terms in \(m_A\) and \(m_B\). Consequently, under either condition, \(E_n^A = o_P(n^{-1/2})\) and \(E_n^B = o_P(n^{-1/2})\).

  Asymptotic Normality via the Delta Method. We have established that both \(\hat{A}\) and \(\hat{B}\) are RAL estimators for \(A_0\) and \(B_0\)
  \begin{align*}
    \hat{A} - A_0 & = P_n[\phi_A(O)] + o_P(n^{-1/2}), \\
    \hat{B} - B_0 & = P_n[\phi_B(O)] + o_P(n^{-1/2}).
  \end{align*}
  By Assumption \ref{ass:regularity_eif}(b), the global variance is non-degenerate, \(B_0 > 0\). Applying the functional Delta method to the ratio \(\hat{\tau}_{\omega}^* = \hat{A}/\hat{B}\) yields
  \begin{align*}
    \hat{\tau}_{\omega}^* - \tau_{\omega}^* & = \frac{1}{B_0}(\hat{A} - A_0) - \frac{A_0}{B_0^2}(\hat{B} - B_0) + o_P(n^{-1/2})   \\
                                            & = P_n\left[\frac{\phi_A(O) - \tau_{\omega}^*\phi_B(O)}{B_0}\right] + o_P(n^{-1/2}).
  \end{align*}
  Substituting the exact expressions for the parameter influence functions
  \begin{align*}
    \phi_A(O) - \tau_{\omega}^*\phi_B(O) & = (m_A(O; \eta_0) - A_0) - \tau_{\omega}^*(m_B(O; \eta_0) - B_0)               \\
                                         & = m_A(O; \eta_0) - \tau_{\omega}^*m_B(O; \eta_0) - (A_0 - \tau_{\omega}^*B_0).
  \end{align*}
  Since \(\tau_{\omega}^* = A_0/B_0\), the constant term exactly cancels. Recalling the unscaled augmented score \(\psi_{\mathrm{aug}}\) from Theorem \ref{thm:eif_optimal}, the scaled influence function identically simplifies to \(B_0^{-1}\psi_{\mathrm{aug}}(O; \tau_{\omega}^*, \eta_0)\). By the Central Limit Theorem, \(\sqrt{n}(\hat{\tau}_{\omega}^* - \tau_{\omega}^*) \xrightarrow{d} \mathcal{N}(0, V_{\omega}^*)\). This completes the proof.
\end{proof}

\newpage
\section{Appendix: Extra Results}\label{sec:appendb}

\subsection{Generalization to Categorical Treatments}

We now demonstrate how our continuous balancing weights framework incorporates existing results for categorical treatments \(D_i \in \{1,\dots,J\}\) as special cases. In this discrete setting, the analytical focus shifts from a continuous derivative to pairwise contrasts defined relative to a target population specified by a marginal tilting function \(w(X_i)\), normalized such that \(\E[w(X_i)] = 1\).

Let \(e_j(x) = \mathbb{P}(D_i=j|X_i=x)\) be the generalized propensity score and let \(\sigma^2(j,x) = \Var(Y_i \mid D_i=j, X_i=x)\) denote the conditional outcome variance. The weighted pairwise contrast between treatment \(j\) and \(k\) is
\[
  \tau_{jk}(w) = \E[w(X_i)(\mu(j, X_i) - \mu(k, X_i))].
\]

The discrete Riesz Representer (balancing weight) for \(\tau_{jk}(w)\) that satisfies the covariate balance constraint from Definition \ref{def:class} is
\[
  \alpha_{jk}(D_i,X_i; w) = w(X_i) \left( \frac{\mathbb{I}(D_i=j)}{e_j(X_i)} - \frac{\mathbb{I}(D_i=k)}{e_k(X_i)} \right).
\]

Following Equation \eqref{eq:vs}, the corresponding conditional efficiency bound \(V_{S, jk}(w)\) is
\[
  \E[\alpha_{jk}^2(D_i,X_i; w)\sigma^2(D_i,X_i)].
\]
We evaluate this expectation by taking the iterated expectation over the discrete support of \(D_i\)
\begin{align*}
  V_{S,jk}(w) & = \E \left[ \sum_{l=1}^J e_l(X_i) \alpha_{jk}^2(l, X_i; w) \sigma^2(l,X_i) \right]                                                    \\
              & = \E \left[ w^2(X_i) \left( e_j(X_i) \frac{\sigma^2(j,X_i)}{e_j^2(X_i)} + e_k(X_i) \frac{\sigma^2(k,X_i)}{e_k^2(X_i)} \right) \right] \\
              & = \E \left[ w^2(X_i) \left( \frac{\sigma^2(j,X_i)}{e_j(X_i)} + \frac{\sigma^2(k,X_i)}{e_k(X_i)} \right) \right].
\end{align*}

\paragraph{With Multiple Treatments (\(J \ge 3\)).} When comparing multiple nominal treatments, \citet{li2019propensity} proposed selecting the target population by minimizing the total asymptotic variance of all possible pairwise contrasts. We adopt this criterion, minimizing this total variance subject to \(\E[w(X_i)] = 1\).
\[
  \sum_{1 \le k < j \le J} V_{S, jk}(w) = \E \left[ w^2(X_i) \sum_{1 \le k < j \le J} \left( \frac{\sigma^2(j,X_i)}{e_j(X_i)} + \frac{\sigma^2(k,X_i)}{e_k(X_i)} \right) \right].
\]
In the summation over all unique pairs, each specific term \(\sigma^2(l,X_i)/e_l(X_i)\) appears exactly \(J-1\) times. Thus, the objective simplifies to
\[
  \sum_{1 \le k < j \le J} V_{S, jk}(w) = (J-1) \E \left[ w^2(X_i) \sum_{j=1}^J \frac{\sigma^2(j,X_i)}{e_j(X_i)} \right].
\]
Let \(\kappa(X_i) = \sum_{j=1}^J \sigma^2(j,X_i)/e_j(X_i)\). We seek to minimize \(\E[w^2(X_i)\kappa(X_i)]\) subject to \(\E[w(X_i)] = 1\). By the Cauchy-Schwarz inequality or a standard Lagrangian approach, the optimal variance-minimizing weight is inversely proportional to \(\kappa(x)\)
\[
  w^*(x) \propto \left( \sum_{j=1}^J \frac{\sigma^2(j,x)}{e_j(x)} \right)^{-1}.
\]
Under homoskedasticity (\(\sigma^2(j,x)=\sigma^2\)), the outcome variance terms cancel out, and the optimal weight becomes the harmonic mean of the generalized propensity scores, \(w^*(x) \propto \left( \sum_{j=1}^J 1/e_j(x) \right)^{-1}\), which perfectly recovers the Generalized Overlap Weights of \citet{li2019propensity}.

\paragraph{With Binary Treatment (\(J = 2\)).} If \(D_i \in \{0, 1\}\), there is only one relevant contrast, \(\tau_{10}\). The optimization criterion naturally reduces to minimizing \(V_{S, 10}(w)\). Following the derivation above, the optimal target weight is
\[
  w^*(x) \propto \left( \frac{\sigma^2(1,x)}{e_1(x)} + \frac{\sigma^2(0,x)}{e_0(x)} \right)^{-1}.
\]
Under homoskedasticity, this simplifies to \(w^*(x) \propto (1/e_1(x) + 1/e_0(x))^{-1} = e_1(x)e_0(x)\). Since \(e_1(x) + e_0(x) = 1\), this algebraic equivalence exactly recovers the optimal overlap weights formally established by \citet{crump2006moving} and \citet{li2018balancing}.

\subsection{Generalization of Theorem 5.1 of Crump et al. (2006)}

When estimating an average derivative effect \(\tau_w\) based on a pre-specified derivative weight \(w(d,x)\), the asymptotic variance of an efficient estimator \(V_\phi\) generally decomposes as \(V_\phi = V_\psi + V_{\mathrm{adj}} + 2C_{\psi, \mathrm{adj}}\). Here, \(V_\psi\) is the variance assuming the true weights are known, \(V_{\mathrm{adj}}\) is the variance inflation from estimating the unknown nuisance parameters within the weights, and \(C_{\psi, \mathrm{adj}}\) is their covariance.

In certain specific settings, structural orthogonality conditions guarantee that the covariance term \(C_{\psi, \mathrm{adj}}\) is exactly zero. This leads to a simpler, additive decomposition (\(V_\phi = V_\psi + V_{\mathrm{adj}}\)), implying that evaluating the estimand using estimated weights increases the asymptotic variance (\(V_\phi \ge V_\psi\)).

We present a generalized corollary that isolates the exact mathematical conditions required for this simplification to hold, formalizing the findings of Theorem 5.1 in \citet{crump2006moving}. Let \(Z_i = (D_i, X_i)\).
\begin{corollary}
  \label{cor:crump}
  Assume that the weight function is defined as \(w(Z_i) = \lambda(\eta(W_i))\), where \(\lambda(\cdot)\) is a known continuously differentiable mapping, \(W_i\) is a subvector of \(Z_i\), and \(\eta(w) = \E[H(O_i) \mid W_i=w]\) for some measurable function \(H(O_i) \in L_2(P)\). Furthermore, assume the following two orthogonality conditions hold:
  \begin{enumerate}
    \item \(\E[(Y_i - \mu(Z_i))(H(O_i) - \eta(W_i)) \mid Z_i] = 0\).
    \item \(\E\big[w(Z_i)(\mu'(Z_i) - \tau_w) \, \E[\lambda'(\eta(W_i))(\mu'(Z_i) - \tau_w) \mid W_i] \, (H(O_i) - \eta(W_i))\big] = 0\).
  \end{enumerate}
  Under these conditions, the covariance between the fixed-weight influence function \(\psi(O_i)\) and the pathwise adjustment term \(\phi_{\mathrm{adj}}(O_i)\) is identically zero (\(C_{\psi, \mathrm{adj}} = 0\)). Assuming the unnormalized weights are scaled by \(K = \E[w(Z_i)]\), the efficiency bound decomposes as \(V_\phi = V_\psi + V_{\mathrm{adj}}\), where
  \[
    V_{\mathrm{adj}} = \frac{1}{K^2} \E \left[ \left( \E[\lambda'(\eta(W_i)) (\mu'(Z_i) - \tau_w) \mid W_i] \right)^2 \Var(H(O_i) \mid W_i) \right].
  \]
\end{corollary}

Corollary \ref{cor:crump} explains why the simple additive variance penalty conventionally holds in the binary treatment setting. In the binary case, target populations depend purely on the propensity score \(\eta(X_i) = \mathbb{P}(D_i=1 \mid X_i)\). Therefore, \(W_i = X_i\) and \(H(O_i) = D_i\). Condition 1 is trivially satisfied because \((D_i - e(X_i))\) is deterministic given \(Z_i\), yielding \(\E[Y_i - \mu(Z_i) \mid Z_i] = 0\). Condition 2 is satisfied because, by the definition of \(e(X_i)\), \(\E[D_i - e(X_i) \mid X_i] = 0\).

This corollary also explains why optimal continuous treatments under general heteroskedasticity preclude this simple additive decomposition. The continuous optimal balancing weights rely on estimating the precision weight \(\omega(Z_i)\), which sets \(W_i = Z_i\) and \(H(O_i) = (Y_i - \mu(Z_i))^2\). In this scenario, Condition 1 evaluates the third conditional moment of the outcome residual
\[
  \E \left[(Y_i - \mu(Z_i))\left((Y_i - \mu(Z_i))^2 - \sigma^2(Z_i)\right) \mid Z_i\right] = \E[(Y_i - \mu(Z_i))^3 \mid Z_i].
\]
Unless the conditional error distribution is symmetric, this conditional covariance is non-zero, violating Condition 1. Furthermore, because the continuous estimand \(\mu'(Z_i)\) depends on \(D_i\), Condition 2 is violated when estimating the continuous propensity score \(e_\omega(X_i)\).

\begin{proof}[Proof of Corollary \ref{cor:crump}]
  Accounting for the normalization constant \(K = \E[w(Z_i)]\), the true influence function assuming fixed weights is
  \[
    \psi(O_i) = \frac{1}{K} \alpha_w(Z_i)(Y_i - \mu(Z_i)) + \frac{w(Z_i)}{K} (\mu'(Z_i) - \tau_w).
  \]

  Under the assumption that the weights depend smoothly on \(\eta(W_i) = \E[H(O_i) \mid W_i]\), standard pathwise differentiation yields the adjustment term
  \[
    \phi_{\mathrm{adj}}(O_i) = \frac{1}{K} \E[\lambda'(\eta(W_i)) (\mu'(Z_i) - \tau_w) \mid W_i] (H(O_i) - \eta(W_i)).
  \]
  Let \(C(W_i) = \frac{1}{K} \E[\lambda'(\eta(W_i)) (\mu'(Z_i) - \tau_w) \mid W_i]\). The covariance \(C_{\psi, \mathrm{adj}} = \E[\psi(O_i)\phi_{\mathrm{adj}}(O_i)]\) expands to
  \begin{align*}
    \E[\psi(O_i)\phi_{\mathrm{adj}}(O_i)] & = \E \left[ \frac{1}{K} \alpha_w(Z_i)(Y_i - \mu(Z_i)) C(W_i)(H(O_i) - \eta(W_i)) \right]    \\
                                          & \quad + \E \left[ \frac{w(Z_i)}{K} (\mu'(Z_i) - \tau_w) C(W_i)(H(O_i) - \eta(W_i)) \right].
  \end{align*}
  By the Law of Iterated Expectations conditioning on \(Z_i\), the first term contains \(\E[(Y_i - \mu(Z_i))(H(O_i) - \eta(W_i)) \mid Z_i]\), which vanishes under Condition 1. The second term vanishes under Condition 2. Consequently, \(C_{\psi, \mathrm{adj}} = 0\). The variance \(V_{\mathrm{adj}} = \E[\phi_{\mathrm{adj}}^2(O_i)]\) follows directly from \(\E[(H(O_i) - \eta(W_i))^2 \mid W_i] = \Var(H(O_i) \mid W_i)\).
\end{proof}

\newpage
\section{Appendix: Replication of HI04 \label{sec:appendc}}
\renewcommand{\thetable}{C\arabic{table}}
\setcounter{table}{0}
\renewcommand{\thefigure}{C\arabic{figure}}
\setcounter{figure}{0}

\begin{table}[h]
\centering
\caption{Summary Statistics for Lottery Winners Sample}
\label{tab:summary}
\begin{tabular}{@{}lcc@{}}
\toprule
Variable & Mean & SD \\
\midrule
\multicolumn{3}{l}{\textit{Outcome and Treatment}} \\
Earnings (Y) & 10.318 & 13.163 \\
Annualized Prize (D) & 55.196 & 61.803 \\
\addlinespace
\multicolumn{3}{l}{\textit{Demographics}} \\
Age at Winning & 46.945 & 13.797 \\
Years High School & 3.603 & 1.071 \\
Years College & 1.367 & 1.601 \\
Male & 0.578 & 0.495 \\
Tickets Bought & 4.570 & 3.282 \\
Working at Winning & 0.802 & 0.400 \\
Year Won & 1986.059 & 1.294 \\
\addlinespace
\multicolumn{3}{l}{\textit{Pre-Treatment Earnings}} \\
Earnings year (-6) & 11.965 & 11.790 \\
Earnings year (-5) & 12.115 & 11.992 \\
Earnings year (-4) & 12.037 & 12.081 \\
Earnings year (-3) & 12.820 & 12.654 \\
Earnings year (-2) & 13.479 & 12.965 \\
Earnings year (-1) & 14.468 & 13.624 \\
\midrule
Observations & \multicolumn{2}{c}{237} \\
\bottomrule
\end{tabular}
\end{table}

\begin{figure}[t]
\centering
\includegraphics[page=1, width=1\textwidth]{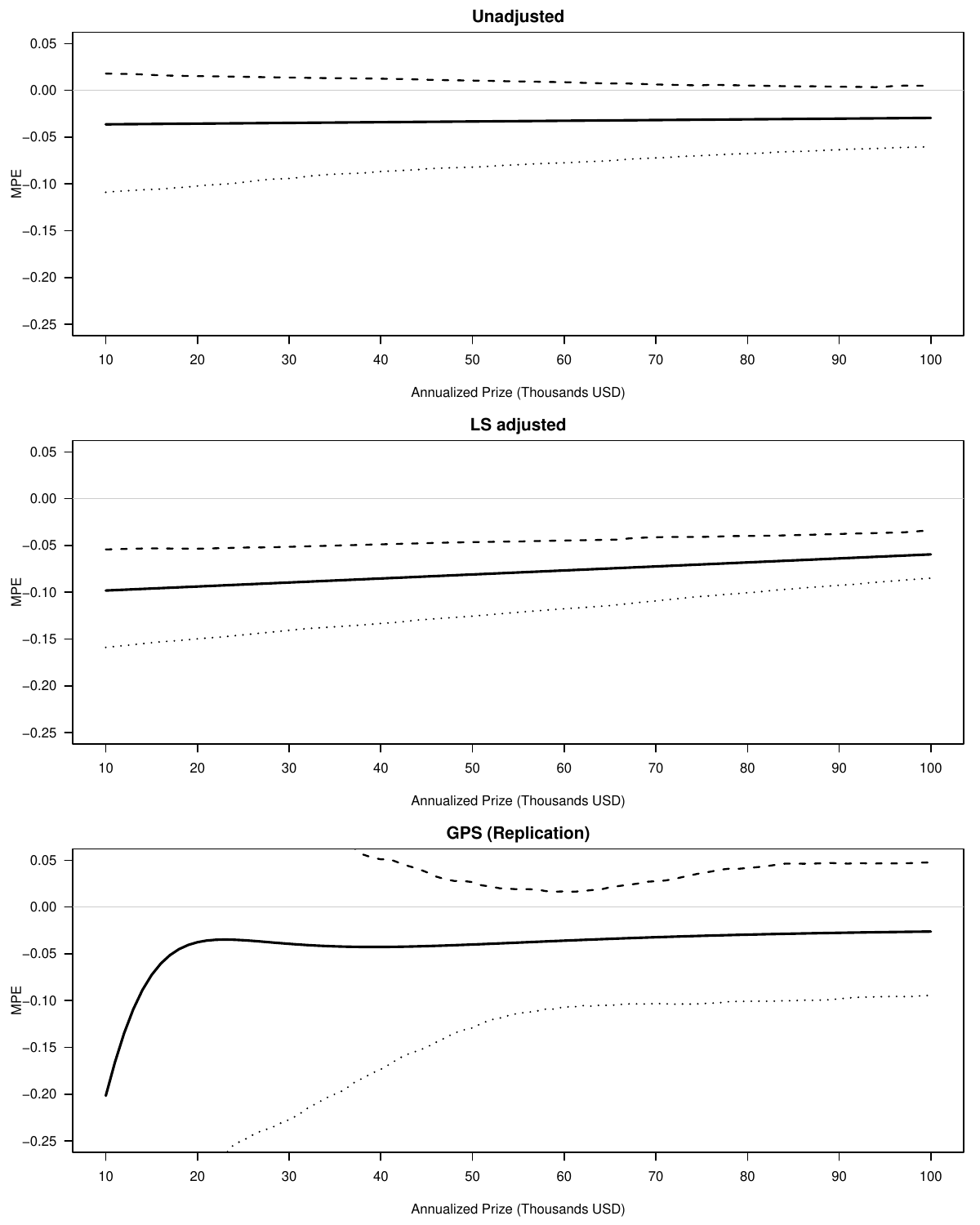}
\caption{Replication of Figure 7.1 in HI04}
\label{fig:mpe_results}
\end{figure}

\end{document}